\documentclass[11pt]{article}

\usepackage[a4paper,margin=1in]{geometry}
\usepackage[T1]{fontenc}
\usepackage[utf8]{inputenc}
\usepackage{lmodern}
\usepackage{amsmath,amssymb,amsthm,mathtools}
\usepackage{enumitem}
\usepackage{graphicx}
\usepackage{caption}
\usepackage[dvipsnames]{xcolor}
\usepackage{hyperref}
\usepackage{aliascnt}
\usepackage[nameinlink,capitalise]{cleveref}

\hypersetup{
  colorlinks=true,
  linkcolor=blue,
  citecolor=blue,
  urlcolor=blue
}

\newtheorem{theorem}{Theorem}[section]

\newaliascnt{proposition}{theorem}
\newtheorem{proposition}[proposition]{Proposition}
\aliascntresetthe{proposition}

\newaliascnt{lemma}{theorem}
\newtheorem{lemma}[lemma]{Lemma}
\aliascntresetthe{lemma}

\newaliascnt{corollary}{theorem}
\newtheorem{corollary}[corollary]{Corollary}
\aliascntresetthe{corollary}

\newaliascnt{conjecture}{theorem}
\newtheorem{conjecture}[conjecture]{Conjecture}
\aliascntresetthe{conjecture}

\theoremstyle{definition}

\newaliascnt{definition}{theorem}
\newtheorem{definition}[definition]{Definition}
\aliascntresetthe{definition}

\newaliascnt{example}{theorem}
\newtheorem{example}[example]{Example}
\aliascntresetthe{example}

\newaliascnt{problem}{theorem}
\newtheorem{problem}[problem]{Problem}
\aliascntresetthe{problem}

\newaliascnt{question}{theorem}

\aliascntresetthe{question}

\theoremstyle{remark}

\newaliascnt{remark}{theorem}
\newtheorem{remark}[remark]{Remark}
\aliascntresetthe{remark}

\newcommand{\F}{\mathbb F}
\newcommand{\Fq}{\mathbb F_q}
\newcommand{\FQ}{\mathbb F_Q}
\newcommand{\PG}{\mathrm{PG}}
\newcommand{\calC}{\mathcal C}
\newcommand{\calK}{\mathcal K}
\newcommand{\calL}{\mathcal L}
\newcommand{\calN}{\mathcal N}
\newcommand{\calS}{\mathcal S}
\newcommand{\calU}{\mathcal U}
\newcommand{\nbf}{\mathbf n}
\newcommand{\abf}{\mathbf a}
\newcommand{\ebf}{\mathbf e}
\newcommand{\ubf}{\mathbf u}
\newcommand{\vbf}{\mathbf v}

\newcommand{\xbf}{\mathbf x}
\newcommand{\cbf}{\mathbf c}
\newcommand{\gbf}{\mathbf g}
\newcommand{\zerobf}{\mathbf 0}
\newcommand{\rk}{\operatorname{rk}}
\newcommand{\wt}{\operatorname{wt}}
\newcommand{\wtH}{\operatorname{wt}_{\mathrm H}}
\newcommand{\wtSR}{\operatorname{wt}_{\mathrm{SR}}}
\newcommand{\distSR}{d_{\mathrm{SR}}}
\newcommand{\distH}{d_{\mathrm H}}
\newcommand{\Norm}{\operatorname{N}}
\def\multiset#1{\ensuremath{\left\{\kern-.3em\left\{{#1}\right\}\kern-.3em\right\}}}

\title{On the Length of Maximum Sum-Rank Distance Codes}
\author{Valentina Astore \thanks{Inria, France. \texttt{valentina.astore@inria.fr}} \thanks{LIX, École polytechnique, Institut Polytechnique de Paris, France.} \thanks{Universit\'e Paris 8, Laboratoire de G\'eom\'etrie, Analyse et Applications, LAGA, Universit\'e Sorbonne Paris Nord, CNRS, UMR 7539, France.}
\and Ferdinando Zullo\thanks{University of Campania \emph{Luigi Vanvitelli}, Caserta, Italy. \texttt{ferdinando.zullo@unicampania.it}}}

\date{}

\begin{document}

\maketitle

\begin{abstract}
The classical theory of arcs in finite projective spaces provides a powerful geometric framework for studying maximum distance separable (MDS) codes and their possible lengths. In this paper, we investigate how this perspective extends to maximum sum-rank distance (MSRD) codes, whose geometric description involves systems of linear sets rather than individual projective points.

Our main goal is to understand how the block structure of an MSRD code constrains its length.
In dimension three, we determine the possible intersection patterns of lines with the associated linear sets and obtain explicit formulas for their characters and the weight distribution of the corresponding codes.
These results lead to new length bounds and show that the block structure comes at a cost: the presence of a block of rank at least two reduces the maximum possible total length by at least $q-1$, and sufficiently long MSRD codes must be classical, that is, MDS codes over $\mathbb F_{q^m}$.
We also prove that three-dimensional one-weight MSRD codes do not exist.

In arbitrary dimension, we interpret shortening geometrically as projection and establish a general projection bound that extends classical counting arguments for arcs.
This yields restrictions depending explicitly on the block profile and, for full-length blocks, identifies parameter ranges in which linearized Reed--Solomon codes attain the maximum possible number of blocks.

Finally, we construct an exceptional three-dimensional MSRD code over $\mathbb F_9$ with three blocks of length two and propose a sum-rank analogue of the MDS conjecture.
\end{abstract}

\\

\tableofcontents

\section{Introduction}

The study of arcs in finite projective spaces is one of the classical problems of finite geometry. An arc is a set of points satisfying a strong independence condition, and understanding how large such sets can be has motivated research for several decades. A fundamental contribution was made by Beniamino Segre in the 1950s. In his celebrated work on ovals in finite projective planes, Segre proved that every oval of $\PG(2,q)$, for $q$ odd, is a non-degenerate conic~\cite{Segre1955Ovals}.
This result established a remarkable connection between a purely combinatorial condition, namely that no three points are collinear, and the algebraic geometry of projective curves. In even characteristic, the situation is substantially different, as conics can be extended to hyperovals. The study of arcs, their maximal sizes, and their classification has since become a central topic in finite geometry, see \cite{abdukhalikov2026some,BallBook,BallLavrauwArcs,Hirschfeld} and the references therein.

A few years later, the introduction of Reed--Solomon codes~\cite{ReedSolomon1960} provided a coding counterpart to these geometric questions.
Indeed, the columns of a generator matrix of a non-degenerate linear code over $\Fq$ define points of a projective space, and the code is \emph{maximum distance separable} (MDS) precisely when these points form an arc.
Under this correspondence, extended Reed--Solomon codes are represented by normal rational curves, while the maximum possible length of an MDS code becomes the maximum possible size of an arc.
This correspondence has proved extremely fruitful, allowing methods and results from finite geometry to be applied to coding theory, and conversely.
Perhaps the most prominent example is the \emph{MDS conjecture}, originating from Segre's investigations, which says that a non-trivial MDS code over $\Fq$ has length at most $q+1$, apart from the well-known exceptional cases associated with hyperovals.
Despite substantial progress, including the proof over prime fields by Ball~\cite{Ball2012} (see also \cite{fan2026linear} for recent results), the conjecture remains open in its full generality.

A different notion of distance emerged with the development of rank-metric codes.
Their foundations go back to Delsarte's work on bilinear forms over finite fields~\cite{Delsarte1978}, while Gabidulin~\cite{Gabidulin1985} subsequently initiated the theory of codes with maximum rank distance, introducing the family now known as Gabidulin codes and an efficient decoding algorithm.
In this setting, the distance between vectors over $\F_{q^m}$ is measured by the dimension of the $\Fq$-span of the entries of their difference, rather than by the number of non-zero coordinates. Codes attaining the corresponding Singleton-like bound are called \emph{maximum rank distance} (MRD) codes, and play a role analogous to that of MDS codes in the Hamming metric.
The importance of rank-metric codes has also been reinforced by applications to network coding, where errors naturally affect the ranks of transmitted matrices~\cite{SilvaKschischangKoetter2008}, and to code-based cryptography.

Although the theory of rank-metric codes initially developed mainly through algebraic methods, its connections with finite geometry have gradually become clearer. In particular, the relationship between MRD codes and scattered subspaces has provided new constructions and new perspectives on the classification of optimal codes, see \textit{e.g.},~\cite{NeriZullo2026,PolverinoZullo2020,Sheekey2016}.
A systematic geometric approach was subsequently developed through the notion of \emph{$q$-systems}~\cite{ALFARANO2022105658,Randrianarisoa2020}.
These are suitable $\Fq$-subspaces of vector spaces over $\F_{q^m}$, closely related to linear sets in finite projective spaces. Under this correspondence, the rank distribution and the minimum distance of a code can be expressed in terms of intersection dimensions with hyperplanes. Thus, the geometric language of projective systems, which has been so successful for Hamming-metric codes, admits a natural counterpart in the rank metric.

A further generalization is provided by the \emph{sum-rank metric}, which simultaneously extends the Hamming and rank metrics.
Here, the coordinates of a vector over $\F_{q^m}$ are partitioned into blocks, and the weight is obtained by summing the $\Fq$-ranks of the individual blocks. Taking all blocks of length one recovers the Hamming metric, whereas a single block gives the rank metric.
The sum-rank metric is particularly relevant in situations where information is transmitted or stored in several groups.
For instance, it arises naturally in multishot network coding, where several transmissions may be affected by rank errors and erasures, see \emph{e.g.}~\cite{MartinezPenasKschischang2019}.
Sum-rank codes also play an important role in distributed storage, particularly in the construction of maximally recoverable locally repairable codes and partial MDS codes~\cite{MartinezPenas2022PMDS}.
These applications have stimulated the development of a general theory of codes endowed with the sum-rank metric, see \cite{byrne2021fundamental,martinez2022codes} for an account of their fundamental properties.

As in the Hamming and rank metrics, a
Singleton-like bound holds in the sum-rank metric.
Codes attaining it are called
\emph{maximum sum-rank distance} (MSRD) codes.
A major step in their construction was made by
Mart\'inez-Pe\~nas, who introduced
\emph{linearized Reed--Solomon codes}
~\cite{MARTINEZPENAS2018587}.
These codes simultaneously generalise
Reed--Solomon and Gabidulin codes and provide
a large family of MSRD codes.
Further constructions, including families based
on extended Moore matrices, have subsequently
been obtained, see \textit{e.g.}, 
\cite{lao2024new,MartinezPenas2022PMDS,martinez2023doubly,MartinezPenas2024,neri2023thegeometry,santonastaso2025onmsrdcodes,santonastaso2023subspace}.

The geometric viewpoint has also been extended
to the sum-rank setting.
In~\cite{neri2023thegeometry}, it was shown
that non-degenerate $\F_{q^m}$-linear sum-rank
metric codes can be described by tuples
of $\Fq$-subspaces, or by
systems of $\Fq$-linear sets in a projective
space. More precisely, a code of dimension $k$
with block profile
$\nbf=(n_1,\ldots,n_t)$
corresponds to a tuple
$\calU=(U_1,\ldots,U_t)$,
where $U_i$ is an $n_i$-dimensional
$\Fq$-subspace of $\F_{q^m}^k$.
The MSRD property is then equivalent to the
condition
\[
\sum_{i=1}^t\dim_{\Fq}(U_i\cap H)\le k-1
\]
for every $\F_{q^m}$-hyperplane $H$.
This provides a natural sum-rank analog
of the geometric characterization of MDS codes
by arcs. In this framework, a geometric description of linearized Reed--Solomon codes based on a $q$-analogue of the normal rational curve has recently been developed in~\cite{astore2026qanalogue}.

The existence of this correspondence raises
a fundamental question:
\emph{to what extent can the classical theory
of arcs and MDS codes be extended to MSRD codes?}
In particular, one would like to understand
which restrictions the MSRD condition imposes
on the possible parameters and, more specifically,
on the length of a code.
Unlike in the Hamming metric, however, the total
length is not the only relevant parameter:
the distribution of the coordinates among
the different blocks plays an essential role.

The main objective of this paper is therefore
to investigate the \emph{length problem} for
MSRD codes from a geometric perspective.
Throughout the paper, we put $Q=q^m$ and
consider an MSRD system of dimension $k$
with profile $\nbf=(n_1,\ldots,n_t)$
and total length $N=n_1+\cdots+n_t$.
A central quantity in our analysis is
\[
M(\nbf)=\sum_{i=1}^t\frac{q^{n_i}-1}{q-1},
\]
which counts the points of the associated
projective linear sets.
While the classical theory of arcs provides
bounds on $N$, the quantity $M(\nbf)$ allows
us to obtain finer restrictions reflecting
the block structure of an MSRD code.

Our first results concern three-dimensional
MSRD codes. In this setting, we determine
the possible intersection sizes of lines
with the associated projective MSRD system
and obtain explicit formulas for its
characters and the weight distribution
of the corresponding code
(\cref{thm:characters,cor:weight-distribution}).
By studying the lines through a point,
we derive restrictions on the block profile
and prove that the presence of a block
of rank at least two reduces the maximum
possible total length by at least $q-1$
(\cref{thm:length-dim3}).
These results yield structural constraints for sufficiently long MSRD codes
(\cref{cor:extremal-dim3,thm:many-blocks}).
As a consequence, MSRD codes that are sufficiently long, or that have sufficiently many blocks, must be classical, that is, MDS codes over $\FQ$ (\cref{cor:extremal-dim3,thm:many-blocks}).
We also establish the non-existence
of three-dimensional one-weight MSRD codes
(\cref{thm:no-one-weight}).

In arbitrary dimension, we investigate
the behaviour of MSRD systems under
shortening and show that this operation
admits a geometric interpretation as
projection. This leads to a general
projection bound (\cref{thm:projection}),
which extends the classical counting
of tangent hyperplanes to an arc.
We use this bound to obtain restrictions
depending explicitly on the block profile
(\cref{cor:profile,cor:length-general}).
For equal block lengths, our approach
recovers a bound of Byrne,
Gluesing-Luerssen and Ravagnani
~\cite{byrne2021fundamental}.
For full-length blocks, we determine
parameter ranges in which the bound
$t\le q-1$ holds and is attained
by linearized Reed--Solomon codes
(\cref{cor:full-blocks}).
We also provide an explicit exceptional
construction over $\F_9$
(\cref{thm:sporadic}) and formulate
a sum-rank analogue of the MDS conjecture
(\cref{conj:msrd}).

\smallskip

The paper is organized as follows.
\Cref{sec:preliminaries} recalls the
necessary background on arcs, MDS codes,
sum-rank metric codes, and their geometric
description.
In \cref{sec:3}, we show the first bound on the parameters of an MSRD code.
In \cref{sec:dim3}, we investigate
three-dimensional MSRD systems,
their intersection distributions,
and the resulting length bounds.
\Cref{sec:one-weight} is devoted
to one-weight MSRD codes in dimension three.
In \cref{sec:general}, we develop
the projection method, derive bounds
in arbitrary dimension, study the case
of equal and full-length blocks,
and propose an MSRD analogue of the
MDS conjecture.
Finally, \cref{sec:conclusions}
summarises our results and discusses
directions for further research.

\section{Preliminaries}\label{sec:preliminaries}

Throughout the paper, $q$ is a prime power, $\Fq$ is the finite field of order $q$, $k$ and $m\geq2$ are integers, $Q:=q^m$, and $\FQ$ is the extension of $\Fq$ of degree $m$. For integers $0\le l\le k$, the \emph{Gaussian binomial coefficient}
\[
\binom{k}{l}_q:=\prod_{i=0}^{l-1}\frac{q^k-q^i}{q^l-q^i}
\]
is the number of $l$-dimensional subspaces of $\Fq^k$. For a non-negative integer $n$ we put
\[
\theta_n:=\binom{n+1}{1}_q=\frac{q^{n+1}-1}{q-1}=1+q+\cdots+q^n,\qquad \theta_{-1}:=0,
\]
so that $\theta_n$ is the number of one-dimensional subspaces of $\Fq^{n+1}$. Note that $\theta_r$ always refers to the subfield $\Fq$, also when we work in projective spaces over $\FQ$. We will often use the identity
\begin{equation}\label{eq:theta-difference}
\theta_{n-1}-\theta_{n-a-1}=q^{n-a}+q^{n-a+1}+\cdots+q^{n-1}=q^{n-a}\,\theta_{a-1},\qquad 0\le a\le n,
\end{equation}
and the inequality
\begin{equation}\label{eq:theta-lower}
\theta_{n-1}\geq n+(q-1)\max\{n-1,0\},
\end{equation}
which holds with equality if and only if $n\le2$. Indeed, \eqref{eq:theta-lower} is an equality for $n\in\{0,1\}$, and for $n\geq2$ we have $\theta_{n-1}=1+\sum_{j=1}^{n-1}q^j\geq1+(n-1)q$, with equality if and only if $n=2$.

Vectors are denoted by bold letters. For a field $\F$ and a vector $\xbf\in\F^k$ we put $\xbf^\perp:=\{\vbf\in\F^k:\xbf\cdot\vbf=0\}$, where $\xbf\cdot\vbf$ is the standard dot product. For integer vectors $\abf=(a_1,\ldots,a_t)$ and $\nbf=(n_1,\ldots,n_t)$, we write $\abf\le\nbf$ if $a_i\le n_i$ for every $i$, we put $|\abf|:=a_1+\cdots+a_t$, and we denote by $\ebf_i$ the $i$-th standard unit vector.

\subsection{Projective spaces and arcs}

The \emph{projective space} $\PG(k-1,q)$ is the set of one-dimensional subspaces of $\Fq^k$. For a vector $\vbf\in\Fq^k\setminus\{\zerobf\}$ we denote by $\langle\vbf\rangle_{\Fq}$ the corresponding point in $\PG(k-1,q)$. A \emph{projective subspace} of dimension $l-1$ of $\PG(k-1,q)$ is the set of points contained in $\PG(W)$ for an $l$-dimensional subspace $W$ of $\Fq^k$. Projective subspaces of dimension $0$, $1$ and $k-2$ are called \emph{points}, \emph{lines} and \emph{hyperplanes}, respectively. The projective space $\PG(k-1,q)$ contains exactly $\binom{k}{l}_q$ projective subspaces of dimension $l-1$. In particular, $\PG(2,q)$ has $q^2+q+1$ points and $q^2+q+1$ lines, every line contains $q+1$ points, and every point lies on $q+1$ lines. The points $\langle\vbf_1\rangle_{\Fq},\ldots,\langle\vbf_s\rangle_{\Fq}$ are \emph{projectively independent} if $\vbf_1,\ldots,\vbf_s$ are linearly independent.

\begin{definition}
Let $\calS$ be a set of points of $\PG(k-1,q)$. A line $\ell$ of $\PG(k-1,q)$ is \emph{external} to $\calS$ if $|\ell\cap\calS|=0$, \emph{tangent} to $\calS$ if $|\ell\cap\calS|=1$, and \emph{$j$-secant} to $\calS$ if $|\ell\cap\calS|=j\geq2$.
\end{definition}

\begin{definition}\label{def:arc}
Let $k\geq2$. An \emph{$N$-arc} of $\PG(k-1,q)$ is a set $\calK$ of $N$ points of $\PG(k-1,q)$ no $k$ of which lie on a hyperplane. Equivalently, any $k$ points of $\calK$ are projectively independent. An arc is \emph{complete} if it is not properly contained in a larger arc of the same space. A $(q+1)$-arc of $\PG(2,q)$ is called an \emph{oval}, and a $(q+2)$-arc of $\PG(2,q)$ a \emph{hyperoval}.
\end{definition}

\begin{example}\label{ex:nrc}
Let $2\le k\le q+1$. The \emph{normal rational curve}
\[
\calN_k:=\left\{\left\langle(1,\alpha,\alpha^2,\ldots,\alpha^{k-1})\right\rangle_{\Fq}:\alpha\in\Fq\right\}\cup\{\left\langle(0,\ldots,0,1)\right\rangle_{\Fq}\}
\]
is a $(q+1)$-arc of $\PG(k-1,q)$. Indeed, $k$ points of $\calN_k$ corresponding to distinct $\alpha_1,\ldots,\alpha_k\in\Fq$ are independent because the corresponding Vandermonde determinant $\prod_{i<j}(\alpha_j-\alpha_i)$ is non-zero. If one of the $k$ points is $\langle(0,\ldots,0,1)\rangle_{\Fq}$, the same argument applies to the Vandermonde matrix of size $k-1$ obtained by deleting the last row and the last column. In particular, for $k=3$, $\calN_3$ is the conic of equation $X_1^2=X_0X_2$.
\end{example}

The first bounds on the size of arcs can be obtained by elementary counting arguments involving hyperplanes. A key observation is that the hyperplanes containing a fixed set of $k-2$ points of an arc form a pencil of size $q+1$, with each remaining point of the arc lying in exactly one hyperplane of the pencil.
This gives the following classical counting formula (see \textit{e.g.}, \cite{Hirschfeld,BallLavrauwArcs}). We will encounter its sum-rank counterpart when studying projections of MSRD systems.

\begin{lemma}\label{lem:tangent-hyperplanes}
Let $\calK$ be an $N$-arc of $\PG(k-1,q)$ and let $A\subseteq\calK$ with $|A|=k-2$. Then there are exactly
\[
q+k-1-N
\]
hyperplanes $H$ of $\PG(k-1,q)$ such that $H\cap\calK=A$.
\end{lemma}
\begin{proof}
The hyperplanes containing $\langle A\rangle$ form a pencil of size $q+1$.
Each point $P\in\calK\setminus A$ determines a unique hyperplane containing $A\cup\{P\}$.
Moreover, distinct points of $\calK\setminus A$ determine distinct hyperplanes, since no hyperplane contains $k$ points of $\calK$.
Consequently, exactly
\[
(q+1)-(N-k+2)=q+k-1-N
\]
hyperplanes meet $\calK$ precisely in $A$.
\end{proof}

Since the number of the hyperplanes considered in the latter lemma is non-negative, we immediately obtain the following elementary upper bound on the size of an arc.
\begin{corollary}\label{cor:arc-trivial-bound}
If $\calK$ is an $N$-arc of $\PG(k-1,q)$, then $N\le q+k-1$.
\end{corollary}

A stronger restriction can be derived if the dimension of the ambient projective space is sufficiently large with respect to the field size.
The following classical result is due to Bush~\cite{Bush1952}, who studied the corresponding problem in the language of orthogonal arrays.
\begin{theorem}[\cite{Bush1952}]\label{thm:bush}
Let $k\geq q$. If $\calK$ is an $N$-arc of $\PG(k-1,q)$, then $N\le k+1$.
\end{theorem}

The case of the projective plane is particularly important and admits a more precise description.
Here, the elementary counting argument derived above already gives $N\le q+2$.
Arcs attaining this bound are called hyperovals and can exist only when $q$ is even.
In odd characteristic, a celebrated theorem of Segre~\cite{Segre1955Ovals} shows that every oval, that is, every $(q+1)$-arc, is a non-degenerate conic.
We collect these classical results in the following theorem (see \cite{BallBook,BallLavrauwArcs,Hirschfeld} for futher details).
\begin{theorem}\label{thm:plane-arcs}
Let $\calK$ be an $N$-arc of $\PG(2,q)$. Then the following hold.
\begin{enumerate}[label=\textup{(\roman*)}]
\item $N\le q+2$, and through every point of $\calK$ there are exactly $q+2-N$ tangent lines to $\calK$.
\item If $N=q+2$, then $q$ is even. Conversely, for $q$ even the conic $\calN_3$ together with the point $\langle(0,1,0)\rangle_{\Fq}$ is a hyperoval.
\item \textup{(Segre)} If $q$ is odd and $N=q+1$, then $\calK$ is a conic, that is, it is projectively equivalent to $\calN_3$.
\end{enumerate}
\end{theorem}

\subsection{Linear codes and MDS codes}

A (linear) \emph{code} of length $N$ and dimension $k$ over $\Fq$ is a $k$-dimensional subspace $\calC\le\Fq^N$. Endowed with the Hamming metric, its \emph{minimum distance} is $\distH(\calC):=\min\{\wtH(\cbf):\cbf\in\calC\setminus\{\zerobf\}\}$, where $\wtH(\cbf)$ is the number of non-zero entries of $\cbf$, and we say that $\calC$ is an $[N,k,\distH(\calC)]_q$ code.

\begin{theorem}[Singleton bound]\label{thm:h-singleton}
If $\calC\le\Fq^N$ is a code of dimension $k$, then $\distH(\calC)\le N-k+1$.
\end{theorem}

Codes attaining the Singleton bound are called \emph{maximum distance separable}, or \emph{MDS}. The dual of an MDS code of dimension $k$ is an MDS code of dimension $N-k$.

A \emph{generator matrix} of $\calC$ is a matrix $G\in\Fq^{k\times N}$ whose rows form a basis of $\calC$, and $\calC$ is \emph{non-degenerate} if no column of $G$ is zero. If $\gbf_1,\ldots,\gbf_N$ are the columns of $G$, the multiset $\multiset{\langle\gbf_1\rangle,\ldots,\langle\gbf_N\rangle}$ of points of $\PG(k-1,q)$ is the \emph{projective system} associated with $\calC$, and it is determined by $\calC$ up to projectivities. Every codeword of $\calC$ is of the form $\xbf G$ with $\xbf\in\Fq^k$, and its $i$-th entry is $\xbf\cdot\gbf_i$, so that
\begin{equation}\label{eq:hamming-weight}
\wtH(\xbf G)=N-\bigl|\{i:\gbf_i\in\xbf^\perp\}\bigr|.
\end{equation}
As $\xbf$ runs over $\Fq^k\setminus\{\zerobf\}$, the subspace $\xbf^\perp$ runs over all hyperplanes of $\Fq^k$. Moreover, each hyperplane arises from exactly $q-1$ choices of $\xbf$. It then follows that $\distH(\calC)$ is equal to $N$ minus the maximum number of columns of $G$ contained in a hyperplane, which yields the following classical characterization.

\begin{theorem}[{\cite[Theorem 2.51 and Corollary 2.52]{Hirschfeld}}] \label{thm:mds-arcs}
Let $\calC$ be a non-degenerate $[N,k]_q$ code with $k\geq2$ and generator matrix $G$. Then $\calC$ is MDS if and only if every $k$ columns of $G$ are linearly independent. Equivalently, $\calC$ is MDS if and only if its projective system is a set of $N$ points forming an $N$-arc of $\PG(k-1,q)$.
\end{theorem}

\begin{example}\label{ex:rs}
Let $\Fq=\{\alpha_1,\ldots,\alpha_q\}$ and $2\le k\le q+1$. The code generated by
\[
G_{\mathrm{RS}}:=
\begin{pmatrix}
1&1&\cdots&1&0\\
\alpha_1&\alpha_2&\cdots&\alpha_q&0\\
\vdots&\vdots&&\vdots&\vdots\\
\alpha_1^{k-2}&\alpha_2^{k-2}&\cdots&\alpha_q^{k-2}&0\\
\alpha_1^{k-1}&\alpha_2^{k-1}&\cdots&\alpha_q^{k-1}&1
\end{pmatrix}
\]
is an \emph{extended Reed--Solomon code} of dimension $k$. Its projective system is the normal rational curve $\calN_k$, hence by \cref{ex:nrc,thm:mds-arcs} it is a $[q+1,k,q-k+2]_q$ MDS code.
\end{example}

\begin{corollary} \label{cor:segre-codes}
Let $q$ be odd and let $\calC$ be a $[q+1,3,q-1]_q$ MDS code. Then the projective system of $\calC$ is a conic, and $\calC$ is monomially equivalent to an extended Reed--Solomon code of dimension $3$.
\end{corollary}
\begin{proof}
By \cref{thm:mds-arcs}, the projective system of $\calC$ is a $(q+1)$-arc of $\PG(2,q)$, hence a conic by \cref{thm:plane-arcs}(iii). Let $A\in\mathrm{GL}(3,q)$ induce a projectivity mapping it onto $\calN_3$. Then the columns of $AG$ are, up to order and non-zero scalars, the columns of $G_{\mathrm{RS}}$, and $AG$ generates $\calC$.
\end{proof}

Further characterizations of Reed--Solomon codes have been explored both geometrically and algebraically in \cite[Castelnuovo's Lemma]{arbarello2013geometry} and \cite[Corollary 27]{mirandola2015critical}.

Despite the many results obtained on arcs and MDS codes, determining their maximum possible length remains a fundamental open problem in finite geometry and coding theory.
The following conjecture, originating from Segre's work, predicts that the normal rational curve essentially provides the largest possible arc, apart from a few exceptional cases.
\begin{conjecture}[MDS conjecture]\label{conj:mds}
Let $2\le k\le q$. If there exists an $[N,k,N-k+1]_q$ MDS code, then $N\le q+1$, unless $q$ is even and $k\in\{3,q-1\}$, in which case $N\le q+2$.
\end{conjecture}

\begin{remark}\label{rem:mds-known}
The exceptional cases of \cref{conj:mds} come from hyperovals ($k=3$) and their duals ($k=q-1$). The conjecture holds for $k\le3$ by \cref{thm:plane-arcs}, hence by duality also for $N-k\le3$. Moreover, it holds for $q$ prime by~\cite{Ball2012}. We refer the reader to \cite{BallLavrauwArcs} for a survey of the topic and to \cite{fan2026linear} for a recent contribution.
\end{remark}

\subsection{Sum-rank metric codes and their geometry}

Let $t\geq1$, let $\nbf=(n_1,\ldots,n_t)$ be a vector of positive integers with $n_i\le m$ for every $i$, and put $N:=n_1+\cdots+n_t$. We identify
\[
\FQ^N\cong\FQ^{n_1}\times\cdots\times\FQ^{n_t}.
\]
For $\cbf=(\cbf_1,\ldots,\cbf_t)\in\FQ^N$ with $\cbf_i\in\FQ^{n_i}$, let $\rk_q(\cbf_i)$ be the dimension of the $\Fq$-span of the entries of $\cbf_i$. Equivalently, $\rk_q(\cbf_i)$ is the rank of the $m\times n_i$ matrix over $\Fq$ obtained by expanding the entries of $\cbf_i$ with respect to a fixed $\Fq$-basis of $\FQ$. The \emph{sum-rank weight} of $\cbf$ is
\[
\wtSR(\cbf):=\sum_{i=1}^t\rk_q(\cbf_i).
\]
For $\nbf=(1,\ldots,1)$ the sum-rank weight is the Hamming weight, and for $t=1$ it reduces to the rank weight.

\begin{definition}
A \emph{sum-rank metric code} with profile $\nbf$ and dimension $k$ is a $k$-dimensional $\FQ$-subspace $\calC\le\FQ^N$, endowed with the sum-rank metric. Its \emph{minimum distance} is $\distSR(\calC):=\min\{\wtSR(\cbf):\cbf\in\calC\setminus\{\zerobf\}\}$, and we say that $\calC$ is an $[\nbf,k,\distSR(\calC)]_{Q/q}$ code.
\end{definition}

\begin{theorem}[Singleton-like bound {\cite[Proposition~16]{MARTINEZPENAS2018587}}]\label{thm:sr-singleton}
If $\calC\le\FQ^N$ is a sum-rank metric code of dimension $k$, then $\distSR(\calC)\le N-k+1$.
\end{theorem}

Codes attaining the bound of \cref{thm:sr-singleton} are called \emph{maximum sum-rank distance}, or \emph{MSRD}. For $\nbf=(1,\ldots,1)$, MSRD codes are exactly the MDS codes over $\FQ$.

\begin{remark}\label{rem:ni-le-m}
In this context, the assumption $n_i\le m$ is standard (see~\cite{byrne2021fundamental}). It is natural because the entries of $\cbf_i$ lie in the $m$-dimensional $\Fq$-space $\FQ$, hence $\rk_q(\cbf_i)\le\min\{n_i,m\}$, and a block of length $n_i>m$ can never have full rank. For a single block the assumption is automatic, that is, if $t=1$ and $\calC$ is an MSRD code with $k<N$, then $N\le m$. Indeed, suppose that $N>m$. Since every codeword has rank at most $m$, we have $d:=N-k+1\le m$. Let $\pi\colon\FQ\to\Fq^{m-d+1}$ be a surjective $\Fq$-linear map, so that $\dim_{\Fq}\ker\pi=d-1$. The $\Fq$-linear map
\[
\calC\longrightarrow\bigl(\Fq^{m-d+1}\bigr)^N,\qquad (c_1,\ldots,c_N)\longmapsto\bigl(\pi(c_1),\ldots,\pi(c_N)\bigr)
\]
is injective. Indeed, every element of its kernel has all its entries in $\ker\pi$, hence rank at most $d-1$, and is therefore zero. Comparing $\Fq$-dimensions gives $km\le N(m-d+1)=N(m-N+k)$, that is, $(N-m)(N-k)\le0$, contradicting $N>m$ and $N>k$.
\end{remark}

We express a generator matrix of a sum-rank metric code $\calC$ in block form $G=(G_1\mid\cdots\mid G_t)$ with $G_i\in\FQ^{k\times n_i}$. The following is a well-known class of MSRD codes, which can be seen as the sum-rank metric analog of Reed--Solomon codes (\emph{cf.} \cref{ex:rs}).

\begin{example} \label{ex:lrs}
    Linearized Reed--Solomon codes were introduced in~\cite{martinez2018skew} as a special family of evaluation codes defined using skew polynomials. For the purposes of this paper, we introduce them directly via their generator matrices.
    
    Let $(\mathbf{a},\mathbf{b})\in\FQ^t\times\FQ^\nbf$ be such that the norms over $\Fq$ of the components of $\mathbf{a}$ are non-zero and pairwise distinct, and the elements of each component of $\mathbf{b}$ are linearly independent over $\Fq$. Let $\sigma$ be a generator of $\mathrm{Gal}(\FQ,\Fq)$ and denote by $\Norm_i(a)=a\cdot \sigma(a)\cdots\sigma^{i-1}(a)$, for every $i\geq 1$ and $a \in \FQ$. The code generated by $G=(G_1|\cdots|G_t)$, where
    \[
    G_i =\left(\begin{array}{ccc}
    b_{i,1} & \dots & b_{i,n_i}\\
    \sigma(b_{i,1})N_1(a_i) & \cdots & \sigma(b_{i,n_i})N_1(a_i) \\
    \sigma^2(b_{i,1})N_2(a_i) & \cdots & \sigma^2(b_{i,n_i})N_2(a_i) \\
    \vdots & & \vdots \\
    \sigma^{k-1}(b_{i,1})N_{k-1}(a_i) & \cdots & \sigma^{k-1}(b_{i,n_i})N_{k-1}(a_i) \\
    \end{array}\right),
    \] is said to be a $[\nbf,k]_{Q/q}$ \emph{linearized Reed--Solomon code}. By~\cite[Theorem 4]{martinez2018skew}, every linearized Reed--Solomon code is MSRD.

    Similarly to what we have seen above for Reed--Solomon codes (see~\cref{ex:rs}), in~\cite{neri2023thegeometry}, this notion was further generalized to doubly-extended linearized Reed--Solomon codes by extending the evaluation of skew polynomials to the points ``$0$'' and ``$\infty$''. These codes are also MSRD, as shown in~\cite[Theorem 4.6]{neri2023thegeometry}.
\end{example}

We say that $\calC$ is \emph{non-degenerate} if, for every $i$, the columns of $G_i$ are linearly independent over $\Fq$ (note that this notion does not depend on the choice of $G$). In the following, we always assume the codes considered to be non-degenerate. For a study of non-degeneracy for sum-rank metric codes that are not necessarily $\FQ$-linear, see~\cite[Section~3.2]{santonastaso2025invariants}. In this case, the $\Fq$-span $U_i$ of the columns of $G_i$ is an $n_i$-dimensional $\Fq$-subspace of $\FQ^k$, and $U_1,\ldots,U_t$ span $\FQ^k$ over $\FQ$, since $G$ has rank $k$. This motivates the following definition, originally introduced in~\cite{neri2023thegeometry}.

\begin{definition}\label{def:system}
An \emph{$[\nbf,k]_{Q/q}$ system} is a tuple $\calU=(U_1,\ldots,U_t)$, where each $U_i$ is an $n_i$-dimensional $\Fq$-subspace of $\FQ^k$ and $\langle U_1,\ldots,U_t\rangle_{\FQ}=\FQ^k$. We call $\nbf$ the \emph{profile} of $\calU$ and $N$ its \emph{total length}.
\end{definition}

Every non-degenerate $[\nbf,k]_{Q/q}$ code determines an $[\nbf,k]_{Q/q}$ system as above, and conversely, choosing $\Fq$-bases of the $U_i$ as columns gives a generator matrix of a non-degenerate code (see~\cite{neri2023thegeometry} for the precise correspondence up to equivalence). For $t=1$ one recovers the notion $q$-system used in~\cite{ALFARANO2022105658,Randrianarisoa2020} to describe rank-metric codes.

For an $\Fq$-subspace $U\le\FQ^k$ and an $\FQ$-subspace $W\le\FQ^k$ we put
\[
\wt_U(W):=\dim_{\Fq}(U\cap W),\qquad \wt_\calU(W):=\sum_{i=1}^t\wt_{U_i}(W),
\]
and for a projective subspace $\calS=\PG(W)$ of $\PG(k-1,Q)$ we also write $\wt_U(\calS)$ and $\wt_\calU(\calS)$.

\begin{definition}\label{def:linear-set}
Let $U$ be an $n$-dimensional $\Fq$-subspace of $\FQ^k$. The set
\[
\calL_U:=\{\langle\ubf\rangle_{\FQ}:\ubf\in U\setminus\{\zerobf\}\}\subseteq\PG(k-1,Q)
\]
is the \emph{$\Fq$-linear set} of rank $n$ defined by $U$.
\end{definition}
It is immediate to see that $|\calL_U|\le\theta_{n-1}$. In particular, a linear set $\calL_U$ is \emph{scattered} if $\wt_U(P)\le1$ for every point $P$, or equivalently if $|\calL_U|=\theta_{n-1}$ (see \textit{e.g.}, \cite{POLVERINO20103096}). More generally, following~\cite{csajbok2021generalising}, for $h\in\{1,\ldots,k-1\}$, we say that $\calL_U$ is \emph{$h$-scattered} if $\wt_U(\calS)\le h$ for every $(h-1)$-dimensional projective subspace $\calS$ of $\PG(k-1,Q)$. In particular, $(k-1)$-scattered linear sets are also called \emph{scattered with respect to hyperplanes}.

\begin{theorem}[{\cite[Theorem~3.1 and Corollary~3.1]{neri2023thegeometry}}]\label{prop:msrd-geometric}
Let $\calC$ be a non-degenerate $[\nbf,k]_{Q/q}$ code with generator matrix $G$ and associated system $\calU$. For every $\xbf\in\FQ^k\setminus\{\zerobf\}$,
\begin{equation}\label{eq:sr-weight}
\wtSR(\xbf G)=N-\wt_\calU(\xbf^\perp).
\end{equation}
Consequently, $\distSR(\calC)=N-\max_H\wt_\calU(H)$, where $H$ runs over the hyperplanes of $\FQ^k$, and $\calC$ is MSRD if and only if
\begin{equation}\label{eq:MSRD}
\wt_\calU(H)\le k-1\qquad\text{for every hyperplane }H\text{ of }\FQ^k.
\end{equation}
\end{theorem}

\begin{definition}\label{def:msrd-system}
An $[\nbf,k]_{Q/q}$ system satisfying~\eqref{eq:MSRD} is an \emph{MSRD system}, and $(\calL_{U_1},\ldots,\calL_{U_t})$ is its associated \emph{projective MSRD system}. We put
\[
M(\nbf):=\sum_{i=1}^t\theta_{n_i-1},\qquad r(\nbf):=\min_i n_i.
\]
The system is \emph{classical} if $\nbf=(1,\ldots,1)$, and \emph{non-trivial} if $2\le k\le N-2$. Finally, we put $\kappa:=\min\{k,N-k\}$.
\end{definition}

The MSRD system associated with the class of linearized Reed--Solomon is described in the following example.

\begin{example}\label{ex:geometry_lrs}
    The structure of the projective MSRD system $(\calL_{U_1},\ldots,\calL_{U_t})$ associated with (doubly-extended) linearized Reed--Solomon codes was studied in~\cite{astore2026qanalogue,neri2023thegeometry}. In particular, \cite[Theorem 2.8]{astore2026qanalogue} shows that the linear sets $\calL_{U_1},\ldots,\calL_{U_t}$ all lie on some $q$-analogue of the normal rational curve $\calN_k$. In particular, if the code is doubly-extended, then its profile consists of the $q+1$ blocks $(m,\dots,m,1,1)$, and the linear sets corresponding with the last two ones are precisely the points zero and infinity in $\PG(k-1,Q)$.
\end{example}

\section{First properties of the MSRD condition}\label{sec:3}

The following lemma establishes the subspace design property of MSRD systems (see~\cite[Proposition~3.2]{santonastaso2023subspace}). For completeness, we provide a short proof using the notation adopted in this paper.
\begin{lemma}\label{lem:design}
Let $\calU$ be an $[\nbf,k]_{Q/q}$ MSRD system. Then
\[
\wt_\calU(W)\le\dim_{\FQ}W
\]
for every $\FQ$-subspace $W\le\FQ^k$ with $\dim_{\FQ}W\le k-1$.
\end{lemma}
\begin{proof}
Let $j=\dim_{\FQ}W$ and put $W_j=W$. If $W_s\neq\FQ^k$ for some $s$, then, since $U_1,\ldots,U_t$ span $\FQ^k$, there exist an index $i$ and a vector $\ubf\in U_i\setminus W_s$. Put $W_{s+1}=W_s+\langle\ubf\rangle_{\FQ}$. Then $\dim_{\FQ}W_{s+1}=s+1$ and $U_i\cap W_{s+1}\supsetneq U_i\cap W_s$, so that $\wt_\calU(W_{s+1})\geq\wt_\calU(W_s)+1$. After $k-1-j$ steps we reach a hyperplane $W_{k-1}$ with $\wt_\calU(W_{k-1})\geq\wt_\calU(W)+k-1-j$. By~\eqref{eq:MSRD}, the left-hand side is at most $k-1$, and hence $\wt_\calU(W)\le j$.
\end{proof}

\begin{corollary}\label{cor:scattered-disjoint}
Let $\calU=(U_1,\ldots,U_t)$ be an $[\nbf,k]_{Q/q}$ MSRD system. Then every $\calL_{U_i}$ is $h$-scattered for every $h\in\{1,\ldots,k-1\}$, in particular, it is scattered and scattered with respect to hyperplanes. Moreover, the linear sets $\calL_{U_1},\ldots,\calL_{U_t}$ are pairwise disjoint, so that
\[
|\calL_{U_1}\cup\cdots\cup\calL_{U_t}|=M(\nbf)=\sum_{i=1}^t\theta_{n_i-1}\le\frac{Q^k-1}{Q-1}.
\]
\end{corollary}
\begin{proof}
By \cref{lem:design}, $\wt_{U_i}(W)\le\wt_\calU(W)\le\dim_{\FQ}W$ whenever $\dim_{\FQ}W\le k-1$, which gives the first statement. For a point $P$, \cref{lem:design} gives $\sum_i\wt_{U_i}(P)\le1$, so $P$ belongs to at most one of the linear sets. The last statement follows by scatteredness since $|\calL_{U_i}|=\theta_{n_i-1}$.
\end{proof}

\begin{remark}\label{rem:converse}
The converse of \cref{cor:scattered-disjoint} does not hold. Consider for instance four points of $\PG(2,Q)$, three of which lie on a line $\ell$ and the fourth outside $\ell$, and let $\calU$ be the $[(1,1,1,1),3]_{Q/q}$ system whose blocks are spanned by vectors representing them. The four linear sets are single points, hence they are scattered, scattered with respect to lines and pairwise disjoint. However $\wt_\calU(\ell)=3$, hence $\calU$ is not MSRD.
\end{remark}

\begin{lemma}\label{lem:msrd-arc}
Let $\calU=(U_1,\ldots,U_t)$ be an $[\nbf,k]_{Q/q}$ MSRD system. For each $i$, let $\ubf_{i,1},\ldots,\ubf_{i,n_i}$ be an $\Fq$-basis of $U_i$. Then the $N$ points $\langle\ubf_{i,j}\rangle_{\FQ}$ are pairwise distinct and form an $N$-arc of $\PG(k-1,Q)$.
\end{lemma}
\begin{proof}
The points are pairwise distinct, because two equal points would have weight at least $2$, contradicting \cref{lem:design}. If $k$ of the vectors $\ubf_{i,j}$ lay on a hyperplane $H$, then those taken from $U_i$ would be $\Fq$-linearly independent elements of $U_i\cap H$, so that $\wt_\calU(H)\geq k$, contradicting~\eqref{eq:MSRD}.
\end{proof}
\begin{remark}
In coding terms, the statement follows from the inequality $\rk_q(\cbf_i)\le\wtH(\cbf_i)$, which holds for every block, together with the two Singleton bounds (\cref{thm:h-singleton} and \cref{thm:sr-singleton}).
\end{remark}

\begin{corollary}\label{cor:classical-transfer}
Every upper bound on the size of arcs of $\PG(k-1,Q)$ holds for the total length of a $k$-dimensional MSRD system. In particular, the following hold for every $[\nbf,k]_{Q/q}$ MSRD system.
\begin{enumerate}[label=\textup{(\roman*)}]
\item $N\le Q+k-1$.
\item If $k=3$, then $N\le Q+2$, and $N\le Q+1$ if $Q$ is odd.
\item If the system is non-trivial and \cref{conj:mds} holds over $\FQ$, then $N\le Q+1$, unless $Q$ is even and $\kappa=3$.
\end{enumerate}
\end{corollary}
\begin{proof}
The first statement follows from \cref{lem:msrd-arc}, and (i) and (ii) then follow from \cref{cor:arc-trivial-bound,thm:plane-arcs}. For (iii), by \cref{lem:msrd-arc} an MSRD code with this system is an $[N,k,N-k+1]_Q$ MDS code with $2\le k\le N-2$. If $k\geq Q$, then \cref{thm:bush} gives $N\le k+1$, a contradiction, so $2\le k\le Q-1$. By \cref{conj:mds}, $N\le Q+1$ unless $Q$ is even, $k\in\{3,Q-1\}$ and $N=Q+2$, in the latter case $\{k,N-k\}=\{3,Q-1\}$ and $Q\geq4$, so $\kappa=3$.
\end{proof}

\begin{remark}\label{rem:conseq_on_t}
The bounds in the preceding corollary also impose constraints on $t$. Recall that $t\le N$, with equality if and only if $n_i=1$ for every $i$, that is, when the system is classical. Consequently, all the bounds above apply to $t$ as well.
\end{remark}

\Cref{cor:classical-transfer} shows that the total length of MSRD codes is controlled by the classical problem on arcs. The sum-rank problem only becomes truly novel when the profile is considered, which is why we often measure MSRD systems by $M(\nbf)$, the number of points of $\calL_{U_1}\cup\cdots\cup\calL_{U_t}$.

In dimension two, hyperplanes are points and the MSRD condition becomes a pure packing condition.
\begin{proposition}\label{prop:dim2}
An $[\nbf,2]_{Q/q}$ system $\calU=(U_1,\ldots,U_t)$ is MSRD if and only if the linear sets $\calL_{U_1},\ldots,\calL_{U_t}$ are scattered and pairwise disjoint in $\PG(1,Q)$. Consequently,
\[
t\le M(\nbf)\le Q+1,
\]
and $Q+1-M(\nbf)$ is the number of points of $\PG(1,Q)$ of weight $0$ in $\calU$. The equality $M(\nbf)=Q+1$ is attained by extended Reed--Solomon codes of dimension $2$ over $\FQ$, with profile $(1,\ldots,1)$ and $t=Q+1$, and by the profile $(m,\ldots,m,1,1)$ with $q-1$ blocks of rank $m$.
\end{proposition}
\begin{proof}
Hyperplanes of $\FQ^2$ are one-dimensional, hence the MSRD condition~\eqref{eq:MSRD} states that $\sum_i\wt_{U_i}(\langle\vbf\rangle_{\FQ})\le1$ for every $\vbf\neq\zerobf$, which means precisely that the linear sets are scattered and pairwise disjoint. Then $|\calL_{U_i}|=\theta_{n_i-1}$, and the bound follows by the fact that the number of points in $\PG(1,Q)$ is $Q+1$.

For the profile $(1,\ldots,1)$, take the $Q+1$ points of $\PG(1,Q)$, that is, the projective system associated with an extended Reed--Solomon code of dimension $2$ over $\FQ$. For the second example, the doubly-extended linearized Reed--Solomon codes in \cref{ex:geometry_lrs} with profile $(m,\ldots,m,1,1)$ attain the upper bound. 
\end{proof}

\section{Three-dimensional MSRD codes}\label{sec:dim3}

Throughout this section, $\calU=(U_1,\ldots,U_t)$ is an $[\nbf,3]_{Q/q}$ MSRD system. To lighten notation we write $\calL_i:=\calL_{U_i}$, and we put
\[
\calL:=\calL_1\cup\cdots\cup\calL_t,\qquad M:=|\calL|.
\]
By \cref{cor:scattered-disjoint}, the linear sets $\calL_i$ are scattered and pairwise disjoint, so that $|\calL_i|=\theta_{n_i-1}$ and $M=M(\nbf)$. Hyperplanes of $\PG(2,Q)$ are lines, and the MSRD condition~\eqref{eq:MSRD} reads
\begin{equation}\label{eq:MSRD3}
\wt_\calU(\ell)=\sum_{i=1}^t\wt_{U_i}(\ell)\le2\qquad\text{for every line }\ell\text{ of }\PG(2,Q).
\end{equation}

The MSRD property in dimension $3$ imposes strong restrictions on the possible intersection sizes of a line with the union $\calL$ of the linear sets forming a projective MSRD system. In particular, only four intersection sizes can occur, in contrast to the unrestricted possibilities for arbitrary point sets in $\PG(2,Q)$. This is formalised in the following lemma.
\begin{lemma}\label{lem:line-intersections}
Let $\ell$ be a line of $\PG(2,Q)$. Then $|\ell\cap\calL|\in\{0,1,2,q+1\}$. More precisely,
\begin{enumerate}[label=\textup{(\roman*)}]
\item $|\ell\cap\calL|=0$ if and only if $\wt_{U_i}(\ell)=0$ for every $i$;
\item $|\ell\cap\calL|=1$ if and only if there is an index $i$ with $\wt_{U_i}(\ell)=1$ and $\wt_{U_j}(\ell)=0$ for every $j\neq i$;
\item $|\ell\cap\calL|=2$ if and only if there are indices $i\neq j$ with $\wt_{U_i}(\ell)=\wt_{U_j}(\ell)=1$ and $\wt_{U_h}(\ell)=0$ for every $h\notin\{i,j\}$;
\item $|\ell\cap\calL|=q+1$ if and only if there is an index $i$ with $\wt_{U_i}(\ell)=2$ and $\wt_{U_j}(\ell)=0$ for every $j\neq i$. In this case, if $\ell=\PG(W)$, then $\ell\cap\calL$ consists of the $q+1$ points defined by the two-dimensional $\Fq$-space $U_i\cap W$.
\end{enumerate}
In particular, $\wt_\calU(\ell)=0$, $1$ or $2$ according to whether $\ell$ is external, tangent, or a $2$- or $(q+1)$-secant to $\calL$.
\end{lemma}
\begin{proof}
Let $\ell=\PG(W)$. For every $i$, the points of $\calL_i$ on $\ell$ are the points $\langle\ubf\rangle_{\FQ}$ with $\ubf\in(U_i\cap W)\setminus\{\zerobf\}$, that is, $\ell\cap\calL_i=\calL_{U_i\cap W}$. Since $\calL_i$ is scattered, so is $\calL_{U_i\cap W}$, and therefore $|\ell\cap\calL_i|=\theta_{\wt_{U_i}(\ell)-1}$. As the linear sets are pairwise disjoint,
\[
|\ell\cap\calL|=\sum_{i=1}^t\theta_{\wt_{U_i}(\ell)-1}.
\]
By the MSRD condition~\eqref{eq:MSRD3}, the weights $\wt_{U_i}(\ell)$ are non-negative integers with sum at most $2$, so either they are all zero, or exactly one of them equals $1$, or exactly two of them equal $1$, or exactly one of them equals $2$, and the others vanish. Since $\theta_{-1}=0$, $\theta_0=1$ and $\theta_1=q+1$, these four cases give $|\ell\cap\calL|=0,1,2,q+1$, respectively. These values are pairwise distinct, and the statement follows.
\end{proof}

For each $j\geq 0$, let $a_j$ denote the number of lines of $\PG(2,Q)$ meeting $\calL$ in exactly $j$ points.
These numbers are called the \emph{characters} of $\calL$.
By \cref{lem:line-intersections}, the only possible non-zero characters are $a_0,a_1,a_2$ and $a_{q+1}$.
The following theorem determines them explicitly in terms of the block profile of the projective MSRD system.
\begin{theorem}\label{thm:characters}
The characters $a_0,a_1,a_2$ and $a_{q+1}$ of $\calL$ satisfy
\begin{gather*}
a_{q+1}=\sum_{i=1}^t\binom{n_i}{2}_q,\qquad
a_2=\sum_{1\le i<j\le t}\theta_{n_i-1}\theta_{n_j-1},\\
a_1=M(Q+1)-2a_2-(q+1)a_{q+1},\qquad
a_0=Q^2+Q+1-M(Q+1)+a_2+q\,a_{q+1}.
\end{gather*}
\end{theorem}
\begin{proof}
Our aim is to show that the following relations hold
\begin{gather*}
a_0+a_1+a_2+a_{q+1}=Q^2+Q+1,\\
a_1+2a_2+(q+1)a_{q+1}=M(Q+1),\\
2a_2+q(q+1)a_{q+1}=M(M-1).
\end{gather*}
By \cref{lem:line-intersections}, every line meets $\calL$ in $0$, $1$, $2$ or $q+1$ points, and $\PG(2,Q)$ has $Q^2+Q+1$ lines, which gives the first identity. The following identities are obtained by standard double-counting arguments. Recall that every point in $\PG(2,Q)$ lies on $Q+1$ lines. Counting the incident pairs $(P,\ell)$ with $P\in\calL$ and $\ell$ a line through $P$, first by lines and then by points, gives the second identity. For the third one, we count the triples $(P,R,\ell)$, where $P,R\in\calL$ are distinct and $\ell$ is the line $PR$. Since a line meeting $\calL$ in $j$ points contributes with $j(j-1)$ triples, while every ordered pair of distinct points determines exactly one line, we conclude.

Consider now the $(q+1)$-secants to $\calL$. If $A\le U_i$ is a two-dimensional $\Fq$-subspace, then $\langle A\rangle_{\FQ}$ is two-dimensional over $\FQ$, otherwise, it would contain a point of weight at least two, contradicting the fact that \(\calL_i\) is scattered. Hence, $\ell_A:=\PG(\langle A\rangle_{\FQ})$ is a line with $\wt_{U_i}(\ell_A)\ge2$. By~\eqref{eq:MSRD3}, equality holds, therefore $U_i\cap\langle A\rangle_{\FQ}=A$, and $\ell_A$ is a $(q+1)$-secant by \cref{lem:line-intersections}(iv). 
Therefore, the only linear set in the considered projective MSRD system having weight two with respect to $\ell_A$ is $\mathcal{L}_i$. 
Conversely, every $(q+1)$-secant to $\calL$ is a $(q+1)$-secant to a unique $\calL_i$. Therefore $a_{q+1}=\sum_i\binom{n_i}{2}_q$.

Next, let $P\in\calL_i$ and $R\in\calL_j$ with $i\neq j$. The line $PR$ has $\wt_{U_i}(PR)\geq1$ and $\wt_{U_j}(PR)\geq1$, so it is a $2$-secant by \cref{lem:line-intersections}, meeting $\calL$ exactly in $\{P,R\}$. Conversely, by \cref{lem:line-intersections}(iii), every $2$-secant meets two distinct linear sets in one point each. Hence the $2$-secants are in bijection with the unordered pairs of points lying in distinct linear sets, and $a_2=\sum_{i<j}\theta_{n_i-1}\theta_{n_j-1}$. The values of $a_1$ and $a_0$ follow from the first two identities.
\end{proof}

\begin{remark}\label{rem:characters}
The third identity in the proof of \cref{thm:characters} is consistent with the explicit values of $a_2$ and $a_{q+1}$, because $q(q+1)\binom{n}{2}_q=\theta_{n-1}(\theta_{n-1}-1)$. If $t=1$, then $a_2=0$ and $a_{q+1}=\binom{n_1}{2}_q$, so every secant line to a single linear set that is scattered with respect to hyperplanes meets it in an $\Fq$-subline. If $n_1=\cdots=n_t=1$, then each $\calL_i$ is a point, $a_{q+1}=0$ and $a_2=\binom{t}{2}$, and we recover the usual character equations of a $t$-arc of $\PG(2,Q)$.
\end{remark}

As a consequence, the weight distribution of a three-dimensional MSRD code only depends on $q$, $m$ and its profile.

\begin{corollary}\label{cor:weight-distribution}
Let $\calC$ be a non-degenerate $[\nbf,3]_{Q/q}$ MSRD code, and for $i \in \{0,\ldots,N\}$ let $w_i$ be the number of codewords of $\calC$ of sum-rank weight $i$. Then $w_0=1$, $w_i=0$ for $i \in \{1,\ldots,N-3\}$, and
\[
w_{N-2}=(Q-1)(a_2+a_{q+1}),\qquad w_{N-1}=(Q-1)\,a_1,\qquad w_N=(Q-1)\,a_0,
\]
where $a_0,a_1,a_2,a_{q+1}$ are the numbers of \cref{thm:characters} for the associated system.
\end{corollary}
\begin{proof}
Let $G$ be a generator matrix of $\calC$ and $\calU$ the associated system. Every non-zero codeword of $\calC$ can be expressed as $\xbf G$ for a unique $\xbf\in\FQ^3\setminus\{\zerobf\}$. By~\eqref{eq:sr-weight} its weight is $N-\wt_\calU(\ell_\xbf)$, where $\ell_\xbf:=\PG(\xbf^\perp)$. Since by the MSRD condition~\eqref{eq:MSRD3} $\wt_\calU(\ell_\xbf)\le2$, all non-zero codewords have weight at least $N-2$. The map $\xbf\mapsto\ell_\xbf$ is $(Q-1)$-to-one onto the set of lines of $\PG(2,Q)$, and by \cref{lem:line-intersections} the codeword $\xbf G$ has weight $N-2$ if $\ell_\xbf$ is a $2$- or a $(q+1)$-secant to $\calL$, weight $N-1$ if $\ell_\xbf$ is tangent to $\calL$, and weight $N$ if $\ell_\xbf$ is external to $\calL$.
\end{proof}

\begin{remark}\label{rem:weight-distribution}
The weight distribution of MSRD codes was previously studied in~\cite[Theorem~VI.5 and Remark~VI.11]{byrne2021fundamental} for matrix codes. \Cref{cor:weight-distribution} explicitly derives the weight distributions in the case of $\FQ$-linear codes of dimension three for arbitrary profiles. For instance, for a hyperoval ($\nbf=(1,\ldots,1)$ and $N=Q+2$) one gets $a_1=0$ and $a_0=Q(Q-1)/2$, hence the only non-zero weights of the code are $Q$ and $Q+2$.
\end{remark}

We now refine the preceding counting arguments by considering the pencil of lines through a fixed point of $\calL$.
The MSRD condition allows us to determine the number of lines of each intersection type in this pencil, leading to a bound on $M$ that depends on the rank of the block containing the chosen point.

\begin{lemma}\label{lem:pencil}
Let $P\in\calL_i$. Among the $Q+1$ lines through $P$:
\begin{enumerate}[label=\textup{(\roman*)}]
\item exactly $\theta_{n_i-2}$ are $(q+1)$-secants to $\calL$, and each of them meets $\calL$ in $q+1$ points of $\calL_i$;
\item exactly $\sum_{j\neq i}\theta_{n_j-1}$ are $2$-secants to $\calL$, and each of them meets $\calL$ in $P$ and in one point of $\calL\setminus\calL_i$;
\item none is external to $\calL$, and the remaining
\[
T_i:=Q+1-\sum_{j\neq i}\theta_{n_j-1}-\theta_{n_i-2}=Q+1-M+q^{n_i-1}
\]
lines are tangent to $\calL$.
\end{enumerate}
In particular, $T_i\geq0$, that is,
\begin{equation}\label{eq:pencil}
\sum_{j\neq i}\theta_{n_j-1}+\theta_{n_i-2}\le Q+1,\qquad\text{or equivalently}\qquad M\le Q+1+q^{n_i-1},
\end{equation}
and the number of tangent lines through $P$ only depends on the rank of the block containing $P$.
\end{lemma}
\begin{proof}
Since $\calL_i$ is scattered, $P=\langle\ubf\rangle_{\FQ}$ for a vector $\ubf\in U_i$ which is unique up to $\Fq^*$-multiples.

(i) If $A\le U_i$ is a two-dimensional $\Fq$-subspace containing $\ubf$, then, as in the proof of \cref{thm:characters}, $\ell_A=\PG(\langle A\rangle_{\FQ})$ is a $(q+1)$-secant through $P$ with $U_i\cap\langle A\rangle_{\FQ}=A$. Conversely, if $\ell=\PG(W)$ is a $(q+1)$-secant through $P$, then by \cref{lem:line-intersections}(iv) there is a unique $j$ such that $\wt_{U_j}(\ell)=2$. Since $\ubf\in U_i\cap W$, we have that $j=i$. So, $A:=U_i\cap W$ is a two-dimensional $\Fq$-subspace of $U_i$ containing $\ubf$ and $\ell=\ell_A$. Therefore the $(q+1)$-secants through $P$ are in bijection with the two-dimensional $\Fq$-subspaces of $U_i$ containing $\langle\ubf\rangle_{\Fq}$, that is, with the one-dimensional subspaces of $U_i/\langle\ubf\rangle_{\Fq}$. Their number is $\theta_{n_i-2}$.

(ii) Let $R\in\calL_j$ with $j\neq i$. By \cref{lem:line-intersections}, the line $PR$ is a $2$-secant to $\calL$, meeting it precisely in $\{P,R\}$. Thus, the map $R\mapsto PR$ defines an injection from $\calL\setminus\calL_i$ into the set of $2$-secants through $P$. To show that the map is also surjective, observe that every $2$-secant through $P$ contains a unique second point $R'\in\calL$. In particular, $R'\notin\calL_i$, since otherwise $\wt_{U_i}(PR')\geq2$, forcing $PR'$ to be a $(q+1)$-secant to $\calL$. Hence, the map is bijective, and the number of $2$-secants through $P$ is $|\calL\setminus\calL_i|=\sum_{j\neq i}\theta_{n_j-1}$.

(iii) Every line through $P$ is non-external and, by \cref{lem:line-intersections}, is either a tangent or one of the lines counted in (i) and (ii). Recalling that the number of lines through a point in $\PG(2,Q)$ is $Q+1$, this yields the first expression for $T_i$. The second follows from $M=\sum_j\theta_{n_j-1}$ and \eqref{eq:theta-difference}.
\end{proof}

Rewriting~\eqref{eq:pencil} in terms of the total length of $\calU$ gives the following inequality, which will play a central role in deriving the next results of this section.
\begin{corollary}\label{cor:defect}
For every $i\in\{1,\ldots,t\}$, put
\[
\Delta_i:=\sum_{j\neq i}\bigl(\theta_{n_j-1}-n_j\bigr)+\bigl(\theta_{n_i-2}-(n_i-1)\bigr).
\]
Then $\Delta_i\geq0$ and
\begin{equation}\label{eq:defect}
N\le Q+2-\Delta_i.
\end{equation}
\end{corollary}
\begin{proof}
Each summand of $\Delta_i$ is non-negative by~\eqref{eq:theta-lower}. Moreover,
\[
\sum_{j\neq i}\theta_{n_j-1}+\theta_{n_i-2}=\sum_{j\neq i}n_j+(n_i-1)+\Delta_i=N-1+\Delta_i,
\]
hence~\eqref{eq:pencil} gives $N-1+\Delta_i\le Q+1$.
\end{proof}

\cref{cor:defect} yield explicit bounds on the total length $N$ of three-dimensional MSRD codes.
In particular, the following theorem shows that the presence of non-trivial blocks imposes stronger restrictions than those arising from the classical theory of arcs.
\begin{theorem}\label{thm:length-dim3}
Let $\calU$ be an $[\nbf,3]_{Q/q}$ MSRD system and put $e:=N-t=\sum_i(n_i-1)$. Then the following hold.
\begin{enumerate}[label=\textup{(\roman*)}]
\item $N\le Q+2$.
\item If $\calU$ is not classical and $t\geq2$, then $N\le Q+3-q$.
\item If $n_i=1$ for some $i$, then $N\le Q+2-(q-1)e$, or equivalently, $t+qe\le Q+2$.
\end{enumerate}
\end{theorem}
\begin{proof}
Part (i) is a direct consequence of \cref{cor:defect}, since $\Delta_i\geq0$. For (ii), let $h$ be an index with $n_h\geq2$ and choose $i\neq h$. All the summands of $\Delta_i$ are non-negative, and the summand corresponding to $h$ satisfies $\theta_{n_h-1}-n_h\geq(q-1)(n_h-1)\geq q-1$ by~\eqref{eq:theta-lower}. Hence $\Delta_i\geq q-1$, and~\eqref{eq:defect} gives $N\le Q+3-q$. For (iii), let $n_i=1$. Then $\theta_{n_i-2}-(n_i-1)=0$ and, by~\eqref{eq:theta-lower}, $\theta_{n_j-1}-n_j\geq(q-1)(n_j-1)$ for every $j\neq i$, so that
\[
\Delta_i\geq(q-1)\sum_{j\neq i}(n_j-1)=(q-1)e.
\]
Therefore,~\eqref{eq:defect} gives $N\le Q+2-(q-1)e$, and substituting $N=t+e$ we conclude that $t+qe\le Q+2$.
\end{proof}

\begin{remark}\label{rem:length-dim3}
\begin{enumerate}[label=\textup{(\alph*)}]
\item Part (i) of \cref{thm:length-dim3} follows from \cref{cor:classical-transfer} with $k=3$. Moreover, the same corollary also gives $N\le Q+1$ when $Q$ is odd.
\item Part (ii) imposes a restriction on the block structure of the system. If $t\geq2$ and $Q+2-N<q-1$, then the system is classical. Thus a single block of rank at least two reduces the maximum
possible total length by at least $q-1$.
\end{enumerate}
\end{remark}
In \cref{sec:general}, we will establish analogous results in arbitrary dimension (see \cref{cor:length-general,rem:length-general} below).

\begin{corollary}\label{cor:extremal-dim3}
Let $\calU$ be an $[\nbf,3]_{Q/q}$ MSRD system and let $\calC$ be an associated code with $\calU$.
\begin{enumerate}[label=\textup{(\roman*)}]
\item If $N=Q+2$, then $\calU$ is classical, $\calL$ is a hyperoval of $\PG(2,Q)$, and $Q$ is even.
\item If $q$ is odd, $t\geq 2$ and $\calU$ is not classical, then $N\le Q$. If moreover $q\geq5$, then $N\le Q-2$. 
\item If $q$ is odd and $N=Q+1$, then $\calU$ is classical and $\calL$ is a conic of $\PG(2,Q)$. In this case, $\calC$ is a $[Q+1,3,Q-1]_Q$ MDS code which is monomially equivalent to an extended Reed--Solomon code of dimension $3$ over $\FQ$.
\end{enumerate}
\end{corollary}
\begin{proof}
If $\calU$ is non-classical, then either $t=1$ and $N\le m<Q$, or $t\geq2$ and $N\le Q+3-q$ by \cref{thm:length-dim3}(ii). In either case, $N\le Q+1$. 
Therefore, if $N=Q+2$ then $\mathcal{U}$ is classical.
In this case, each $\calL_i$ consists of a single point and $t=N$. Moreover, the MSRD condition~\eqref{eq:MSRD3} implies that no line in $\PG(2,Q)$ contains three points of $\calL$. Thus, $\calL$ is an $(Q+2)$-arc, which forces $Q$ to be even by \cref{thm:plane-arcs}(ii). This proves (i).

By \cref{thm:length-dim3}(ii) we have $N\leq Q+3-q$. If $q$ is odd, then $Q+3-q\le Q $, and the tighter bound $Q+3-q\le Q-2$ applies for $q\geq5$. Hence, (ii) follows. 

To prove (iii), suppose that $q$ is odd. Then (ii) implies that $\calU$ is classical. Consequently, $\calL$ is a $(Q+1)$-arc and hence a conic by \cref{thm:plane-arcs}(iii). Since every block has length one, $\calC$ is an MDS code over $\FQ$, and the final statement follows from \cref{cor:segre-codes}.
\end{proof}

The bounds above concern the total length $N$ of an MSRD system. In the following result, we show that the number of blocks also imposes strong restrictions on the structure of the system, forcing it to be classical when the number of blocks is sufficiently large.
\begin{theorem}\label{thm:many-blocks}
Let $\calU$ be an $[\nbf,3]_{Q/q}$ MSRD system with $t\geq Q+2-q$. Then $\calU$ is classical, and $\calL$ is a $t$-arc of $\PG(2,Q)$.
\end{theorem}
\begin{proof}
Suppose that $\calU$ is not classical, so that $e=N-t\geq1$. Note that $t\geq Q+2-q\geq4$. If all blocks had rank at least two, then $N\geq2t\geq2(Q+2-q)>Q+3-q$, contradicting \cref{thm:length-dim3}(ii). Hence some block has rank one, and \cref{thm:length-dim3}(iii) gives $t\le Q+2-qe\le Q+2-q$. Therefore $t=Q+2-q$ and $e=1$.

So there is a unique index $h$ with $n_h=2$, while $n_i=1$ for $i\neq h$. The linear set $\calL_h$ contains the $q+1$ points of an $\Fq$-subline line of $\ell$, and $\ell$ is the only $(q+1)$-secant to $\calL$, since $a_{q+1}=\binom{2}{2}_q=1$ by \cref{thm:characters}. In particular, since by \cref{lem:line-intersections} $\ell$ cannot meet any other point in $\calL$, we have that $\ell\cap\calL=\calL_h$.

Put $\calS:=\calL\setminus\calL_h$, so that $|\calS|=t-1=Q+1-q$ and $M=Q+2$. For every $i\neq h$, \cref{lem:pencil} gives $T_i=Q+1-M+q^{0}=0$, so no point of $\calS$ lies on a tangent line to $\calL$.

Choose a point $P\in\ell\setminus\calL_h$, which exists since $Q>q$. Then $P\notin\calL$. Every line through $P$ other than $\ell$ meets $\ell$ only at $P$ and is therefore disjoint from $\calL_h$. In particular, such a line cannot be a $(q+1)$-secant to $\calL$. By \cref{lem:line-intersections}, it must meet $\calL$ in $0$, $1$ or $2$ points, all belonging to $\calS$. However, it cannot meet $\calL$ in exactly one point, since this would give a tangent to $\calL$ through a point of $\calS$, contradicting the fact that no point of $\calS$ lies on a tangent line. Hence every line through $P$, other than $\ell$, meets $\calS$ in either zero or two points. As $\calS\cap\ell=\emptyset$, the $Q$ lines through $P$ different from $\ell$ partition $\calS$ into subsets of size two, and therefore $|\calS|$ is even. On the other hand, $Q=q^m\equiv q\pmod2$, so $|\calS|=Q+1-q$ is odd, a contradiction.
\end{proof}

\begin{remark}
\Cref{thm:length-dim3}(ii) alone gives $t\le N-1\le Q+2-q$ for non-classical systems with $t\geq2$, that is, $t\geq Q+3-q$ forces a classical system. The parity argument in the proof of \cref{thm:many-blocks} rules out the boundary case $t=Q+2-q$.
\end{remark}

\section{One-weight MSRD codes in dimension three}\label{sec:one-weight}
A sum-rank metric code is called \emph{one-weight} if all its non-zero codewords have the same sum-rank weight. One-weight codes in the sum-rank metric have been studied from a geometric perspective in~\cite{MusharrafZullo2026,neri2023thegeometry}.
In dimension two, one-weight MSRD codes do exist. Indeed, by~\eqref{eq:sr-weight} and \cref{prop:dim2}, a two-dimensional MSRD code is one-weight if and only if every point of $\PG(1,Q)$ has weight one with respect to the projective MSRD system. Equivalently, this holds if the linear sets of its associated projective system are scattered and form a partition of $\PG(1,Q)$, as in the examples given in \cref{prop:dim2}. The existence of one-weight MSRD codes in dimension three was investigated in~\cite{MusharrafZullo2026}, where necessary conditions and partial non-existence results were obtained. In particular, non-existence was proved for certain families of parameters, including the case of full-length blocks over $\F_{2^m}$.
Here, we settle the problem completely by showing that no non-degenerate one-weight MSRD code of dimension three exists, regardless of the field or the block profile.
We use the notation introduced in \cref{sec:dim3}.

\begin{proposition}\label{prop:one-weight-lines}
Let $\calU$ be the system associated with a three-dimensional one-weight MSRD code. Then every line $\ell$ of $\PG(2,Q)$ satisfies $\wt_\calU(\ell)=2$. Consequently, $\ell$ is either a $(q+1)$-secant meeting a single linear set $\calL_i$, or a $2$-secant meeting two distinct linear sets in one point each. In particular, there are no external and no tangent lines to $\calL$.
\end{proposition}
\begin{proof}
The minimum distance of the code is $N-2$, so every non-zero codeword has weight $N-2$. By~\eqref{eq:sr-weight}, it follows that every line has weight exactly $2$ in $\calU$, and the remaining statements follow from \cref{lem:line-intersections}.
\end{proof}

\begin{theorem}\label{thm:one-weight-homogeneous}
Let $\calU$ be the system associated with a three-dimensional one-weight MSRD code. Then $n_1=\cdots=n_t=r$ and
\begin{equation}\label{eq:one-weight-M}
M=Q+1+q^{r-1}.
\end{equation}
\end{theorem}
\begin{proof}
Let $P\in\calL_i$. By \cref{prop:one-weight-lines}, no tangent line to $\calL$ passes through $P$. Thus, $T_i=0$ in \cref{lem:pencil}, which gives $M=Q+1+q^{n_i-1}$. Since $M$ is independent of $i$, then all the $n_i$ must be equal.
\end{proof}

\begin{theorem}\label{thm:no-one-weight}
There exists no non-degenerate three-dimensional one-weight MSRD code over $\FQ$.
\end{theorem}
\begin{proof}
Suppose by contradiction that such a code exists. Let $r=n_1=\dots=n_t$, and recall that we assume $r\leq m$ (see \cref{rem:ni-le-m}). If $t=1$, then
\[
M=\theta_{r-1}\le\theta_{m-1}=\frac{Q-1}{q-1}<Q+1,
\]
contradicting~\eqref{eq:one-weight-M}. Assume therefore that $t\geq2$, and choose $P\in\calL_1$ and $R\in\calL_2$. The line $PR$ is a $2$-secant, so we may choose a point $S\in PR\setminus\calL$. By \cref{prop:one-weight-lines}, each of the $Q+1$ lines through $S$ meets $\calL$ in at least two points, and the intersections between these lines and $\calL$ are pairwise disjoint, since $S\notin\calL$.  Hence $M\geq2(Q+1)$. On the other hand, by~\eqref{eq:one-weight-M} and the fact that $r\le m$,
\[
M=Q+1+q^{r-1}\le Q+1+\frac{Q}{q}<2Q+2,
\]
a contradiction.
\end{proof}

\begin{corollary}\label{cor:no-one-weight}
Let $\calC\le\FQ^{n_1}\times\cdots\times\FQ^{n_t}$, with $n_i\le m$ for every $i$, be a non-degenerate $\FQ$-linear sum-rank metric code of dimension $3$. If $\calC$ is MSRD, then $\calC$ is not one-weight.
\end{corollary}

\begin{remark}\label{rem:one-weight}
\Cref{thm:no-one-weight} highlights a clear contrast between dimensions two and three. In dimension three, requiring every line of $\PG(2,Q)$ to have total weight exactly two is overly restrictive. Note that the last result does not assume that all blocks have the same length, as the homogeneity of the profile follows directly from the code being one-weight (\cref{thm:one-weight-homogeneous}).
\end{remark}

\begin{remark}
    Note that for $t=1$ \cref{cor:no-one-weight} is consistent with~\cite[Corollary 3.17]{ALFARANO2022105658}. Indeed, any (non-degenerate) one-weight rank metric code of dimension $3$ is equivalent to a simplex code (see~\cite{ALFARANO2022105658} for the definition) with parameters $[3m,3,m]_Q$, which is not MSRD.
\end{remark}

\section{MSRD codes of arbitrary dimension}\label{sec:general}
In this section, we extend our investigation of the length problem to MSRD codes of arbitrary dimension. While the results of \cref{sec:dim3} rely on counting lines through a fixed point, in higher dimensions we adopt a more general approach based on shortening and its geometric interpretation as projection. The key idea is to reduce an MSRD system to a two-dimensional one, where the associated linear sets are scattered and pairwise disjoint. The resulting restrictions on their sizes can then be translated into bounds on the original block profile. We begin by studying shortening, puncturing and duality from a geometric perspective. We then establish a general projection bound, derive restrictions on the total length and the number of blocks, and investigate in particular the case of equal and full-length blocks. Finally, we discuss exceptional configurations and propose a sum-rank analogue of the classical MDS conjecture.

\subsection{Shortening, puncturing and duality}
Shortening and puncturing are standard operations on sum-rank metric codes, and it is known that both operations preserve the MSRD property under suitable assumptions (see \cite[Corollary~7]{MartinezPenas2019Supports}). The duality of MSRD codes is also well established (see \cite[Theorem~5]{MartinezPenas2019Supports}). In this subsection, we revisit these operations from the geometric perspective of MSRD systems. In particular, we show that shortening corresponds to projection from a suitable subspace, while puncturing corresponds to replacing the components of the system by appropriate subspaces. These interpretations will be useful for deriving the projection bound in the next subsection.

Shortening and puncturing are two standard operations on linear
codes. In the sum-rank metric, these operations can be performed
after an $\Fq$-linear change of coordinates within each block.
We recall their definitions and describe their geometric
interpretation in terms of systems.

\begin{definition}
\label{def:shortening-puncturing}
Let $\calC$ be an $[\nbf,k]_{Q/q}$ code, and let
$\zerobf\le\abf=(a_1,\ldots,a_t)\le\nbf$. For each $i$, choose $T_i\in\mathrm{GL}(n_i,q)$ and put $T=\operatorname{diag}(T_1,\ldots,T_t)$.
For $\cbf=(\cbf_1,\ldots,\cbf_t)\in\calC$, write
\[
\cbf_iT_i=(\cbf_i',\cbf_i''),
\]
where $\cbf_i'\in\FQ^{a_i}$ and
$\cbf_i''\in\FQ^{n_i-a_i}$.

The \emph{puncturing} of $\calC$ with respect to
$\abf$ and $T$ is the code
\[
\operatorname{Punc}_{\abf,T}(\calC)
:=
\left\{
(\cbf_1'',\ldots,\cbf_t''):
\cbf\in\calC
\right\}.
\]

The \emph{shortening} of $\calC$ with respect to
$\abf$ and $T$ is the code
\[
\operatorname{Short}_{\abf,T}(\calC)
:=
\left\{
(\cbf_1'',\ldots,\cbf_t''):
\cbf\in\calC,\ 
\cbf_i'=\zerobf\text{ for all }i
\right\}.
\]

Both codes have profile $\nbf-\abf$, with blocks
of length zero discarded.
\end{definition}

These operations have natural geometric interpretations
under the correspondence between non-degenerate codes
and systems. Let $G=(G_1\mid\cdots\mid G_t)$ be a
generator matrix of $\calC$, with associated system
$\calU=(U_1,\ldots,U_t)$.
For each $i$, write $G_iT_i=(G_i'\mid G_i'')$,
where $G_i'$ has $a_i$ columns, and let
$A_i$ and $B_i$ denote the $\Fq$-spans of the
columns of $G_i'$ and $G_i''$, respectively.
Then $U_i=A_i\oplus B_i$.

Geometrically, puncturing corresponds to replacing
each $U_i$ by the subspace $B_i\le U_i$.
Thus, whenever the dimension is preserved, the
associated system is
\[
\calU^{\mathrm{punc}}:=(B_1,\ldots,B_t).
\] 

On the other hand, shortening corresponds to a
projection. More precisely, putting $W:=\langle A_1,\ldots,A_t\rangle_{\FQ}$ and assuming that the hypothesis of the next lemma hold,
the associated system of the shortened code, is
\[
\calU^{\mathrm{short}}
:=\frac{\calU}{W}=
\left(
\frac{U_1+W}{W},\ldots,
\frac{U_t+W}{W}
\right).
\] 
In particular, shortening consists of projecting
the system from the projective subspace $\PG(W)$.

We now show that shortening an MSRD system preserves
the MSRD property and is consistent with the geometric
interpretation of shortening presented so far.

\begin{lemma}\label{lem:projection-shortening}
Let $\calU=(U_1,\ldots,U_t)$ be an $[\nbf,k]_{Q/q}$ MSRD system. Let $A_i\le U_i$ be $\Fq$-subspaces with $a_i:=\dim_{\Fq}A_i$ and $j:=|\abf|\le k-2$, and put $W:=\langle A_1,\ldots,A_t\rangle_{\FQ}$. Then the following hold.
\begin{enumerate}[label=\textup{(\roman*)}]
\item $\dim_{\FQ}W=j$ and $U_i\cap W=A_i$ for every $i$.
\item For every hyperplane $H\supseteq W$, we have $\wt_{\calU/W}(H/W)=\wt_\calU(H)-j$. 
\item $\calU/W$ is a $[\nbf-\abf,k-j]_{Q/q}$ MSRD system in $\FQ^k/W\cong\FQ^{k-j}$, where the blocks with $a_i=n_i$ are discarded.
\end{enumerate}
\end{lemma}
\begin{proof}
(i) We have $\dim_{\FQ}W\le j\le k-1$ and $\wt_\calU(W)\geq\sum_ia_i=j$. By \cref{lem:design}, $\wt_\calU(W)\le\dim_{\FQ}W$, hence $\dim_{\FQ}W=j$ and $\dim_{\Fq}(U_i\cap W)=a_i$ for every $i$; since $A_i\le U_i\cap W$, we get $U_i\cap W=A_i$.

(ii) Since $W\le H$, the modular law gives $(U_i+W)\cap H=(U_i\cap H)+W$. Therefore, using (i),
\[
\dim_{\Fq}\bigl(((U_i+W)\cap H)/W\bigr)=\dim_{\Fq}(U_i\cap H)-\dim_{\Fq}(U_i\cap W)=\dim_{\Fq}(U_i\cap H)-a_i,
\]
and summing over $i$ gives the claim.

(iii) The images of the $U_i$ span $\FQ^k/W$, and $\dim_{\Fq}\bigl((U_i+W)/W\bigr)=n_i-a_i$ by (i). Every hyperplane of the quotient is of the form $H/W$ with $H\supseteq W$, so by (ii) it has weight at most $(k-1)-j$, which is the MSRD condition~\eqref{eq:MSRD} in dimension $k-j$.
\end{proof}

Similarly, puncturing also preserves the MSRD property, provided that the resulting total length is sufficiently large. Combining this fact with the duality of MSRD codes, we obtain the following result.
We recall that the \emph{dual code} of a $[\mathbf{n},k]_{Q/q}$ code $\calC$ is 
\[\calC^\perp:=\{\mathbf{v} \in \FQ^{\mathbf{n}} \colon \mathbf{v}\cdot \mathbf{c}=0, \forall \mathbf{c}\in \calC\},\]
which is a $[\mathbf{n},N-k]_{Q/q}$ code.

\begin{lemma}\label{lem:puncture-dual}
Let $\calU$ be an $[\nbf,k]_{Q/q}$ MSRD system with associated code $\calC$.
\begin{enumerate}[label=\textup{(\roman*)}]
\item If $N\geq k+1$ and $U_i'\le U_i$ is an $\Fq$-hyperplane, then replacing $U_i$ by $U_i'$ gives an $[\nbf-\ebf_i,k]_{Q/q}$ MSRD system, where the block is discarded if $n_i=1$.
\item If $k\le N-1$, then $\calC^\perp$ is a non-degenerate MSRD code of dimension $N-k$ with the same profile $\nbf$.
\end{enumerate}
\end{lemma}
\begin{proof}
(i) Weights do not increase, so the MSRD condition~\eqref{eq:MSRD} still holds. If the new tuple was contained in a hyperplane $H$, then its weight on $H$ would be $N-1\geq k$, contradicting~\eqref{eq:MSRD}, hence it spans $\FQ^k$.

(ii) Since $\calC$ is MSRD, it follows immediately from~\cite[Theorem 5]{MartinezPenas2019Supports} that $\calC^\perp$ is also MSRD. We note that $\calC$ and $\calC^\perp$ have the same profile. If a block of a generator matrix of $\calC^\perp$ had $\Fq$-linearly dependent columns, then by \cite[Proposition 2.13]{neri2023thegeometry} $\calC=(\calC^\perp)^\perp$ would contain a codeword of sum-rank weight $1$, whereas $\distSR(\calC)=N-k+1\geq2$.
\end{proof}

\begin{remark}
\label{rem:short-punct-codes}
The preceding lemmas have the following interpretations
in coding terms. Let $\calC$ be a
non-degenerate $[\nbf,k,N-k+1]_{Q/q}$ MSRD code.
By \cref{lem:projection-shortening}, shortening
$\calC$ on $j\le k-2$ coordinates, distributed
among the blocks according to a vector
$\abf\le\nbf$ with $|\abf|=j$, produces a
non-degenerate
$[\nbf-\abf,k-j,N-k+1]_{Q/q}$
MSRD code. In particular, shortening preserves
the minimum sum-rank distance.
On the other hand, by \cref{lem:puncture-dual}(i),
if $N\ge k+1$, puncturing $\calC$ on one coordinate
of the $i$-th block produces a non-degenerate
$[\nbf-\ebf_i,k,N-k]_{Q/q}$
MSRD code. Thus puncturing preserves the dimension
and decreases the minimum sum-rank distance by one.
In both cases, the operations may be performed
after an $\Fq$-linear change of coordinates
within each block.
\end{remark}

\subsection{General bounds}
The geometric interpretation of shortening allows us to derive
bounds on MSRD systems of arbitrary dimension by reducing to
the two-dimensional case. More precisely, projecting an
$[\nbf,k]_{Q/q}$ MSRD system from a subspace generated by
$k-2$ suitably chosen vectors produces a two-dimensional
MSRD system. Applying the characterization of
\cref{prop:dim2} to the projected system yields the following
bound, together with an exact count of the hyperplanes
attaining the smallest possible weight.

\begin{theorem}\label{thm:projection}
Let $\calU$ be an $[\nbf,k]_{Q/q}$ MSRD system. Let $\zerobf\le\abf\le\nbf$ with $|\abf|=k-2$, let $A_i\le U_i$ be any $\Fq$-subspaces with $\dim_{\Fq}A_i=a_i$, and put $W:=\langle A_1,\ldots,A_t\rangle_{\FQ}$. Then every hyperplane through $W$ has weight $k-2$ or $k-1$, and exactly
\[
\tau(A_1,\ldots,A_t):=Q+1-\sum_{i=1}^t\theta_{n_i-a_i-1}
\]
of them have weight $k-2$. In particular,
\begin{equation}\label{eq:projection}
\sum_{i=1}^t\theta_{n_i-a_i-1}\le Q+1.
\end{equation}
If moreover $k\le N-2$, then~\eqref{eq:projection} also holds for every $\zerobf\le\abf\le\nbf$ with $|\abf|=N-k-2$.
\end{theorem}
\begin{proof}
By \cref{lem:projection-shortening} with $j=k-2$, the quotient $\calU/W$ is a two-dimensional MSRD system with profile $\nbf-\abf$. The hyperplanes through $W$ correspond to the points of $\PG(\FQ^k/W)\cong\PG(1,Q)$, with weights shifted by $k-2$. By \cref{prop:dim2}, the points of weight one are the $\sum_i\theta_{n_i-a_i-1}$ points of a disjoint union of scattered linear sets, and the other points have weight zero. The last statement follows by applying the same argument to $\calC^\perp$, which is allowed by \cref{lem:puncture-dual}(ii).

Assume now that $k\le N-2$, and let $\zerobf\le\abf\le\nbf$ with $|\abf|=N-k-2$. Let $\calC$ be a non-degenerate code associated with $\calU$. Since $\calU$ is an MSRD system, $\calC$ is an MSRD code. As $k\le N-1$, by \cref{lem:puncture-dual}(ii) the dual code $\calC^\perp$ is a non-degenerate MSRD code of dimension $k':=N-k$ with the same profile $\nbf$. Let $\calU'=(U_1',\ldots,U_t')$ be a system associated with $\calC^\perp$. By \cref{prop:msrd-geometric}, $\calU'$ is an $[\nbf,k']_{Q/q}$ MSRD system. Since $|\abf|=k'-2$ and $a_i\le n_i=\dim_{\Fq}U_i'$ for every $i$, we can choose $\Fq$-subspaces $A_i'\le U_i'$ with $\dim_{\Fq}A_i'=a_i$, and apply the first part of the statement to $\calU'$ and $A_1',\ldots,A_t'$. This gives~\eqref{eq:projection}, since this inequality depends only on $\nbf$ and $\abf$, and not on the system.
\end{proof}

\begin{remark}\label{rem:projection-special}
\Cref{thm:projection} is the common generalization of three counts that appeared before.
\begin{enumerate}[label=\textup{(\alph*)}]
\item In the classical case, taking $\abf$ to be the vector with $k-2$ entries equal to $1$ and the remaining are zero gives $\tau=Q+k-1-N$, which is \cref{lem:tangent-hyperplanes}.
\item For $k=3$ and $\abf=\ebf_i$, the subspace $W$ is simply a point $P\in\calL_{U_i}$, and by~\eqref{eq:theta-difference}
\[
\tau=Q+1-M(\nbf)+q^{n_i-1},
\]
which is the number $T_i$ of tangent lines through $P$ of \cref{lem:pencil}. The $(q+1)$-secants through $P$ correspond to the $\theta_{n_i-2}$ points of the projection of $\calL_{U_i}$, and the $2$-secants through $P$ to the points of the projections of the other linear sets.
\item For $k=2$, \cref{thm:projection} is \cref{prop:dim2}.
\end{enumerate}
\end{remark}

\begin{remark}
In coding terms, with the notation of \cref{thm:projection}, let $\calC$ be a code associated with $\calU$, with generator matrix $G=(G_1\mid\cdots\mid G_t)$ such that the columns of $G_i$ form an $\Fq$-basis of $U_i$ for every $i$. Since $A_i\le U_i$, we can choose $T_i\in\mathrm{GL}(n_i,q)$ such that the first $a_i$ columns of $G_iT_i$ form an $\Fq$-basis of $A_i$. We put $T:=\operatorname{diag}(T_1,\ldots,T_t)$ and $\calC':=\operatorname{Short}_{\abf,T}(\calC)$. We note that $\calC'$ is a non-degenerate two-dimensional MSRD code of length $N-k+2$, whose associated system is $\calU/W$. By~\eqref{eq:sr-weight} and \cref{lem:projection-shortening}, a non-zero codeword of $\calC'$ has sum-rank weight $N-k+2$ if and only if it corresponds to a point $P$ of $\PG(\FQ^k/W)$ with $\wt_{\calU/W}(P)=0$. Therefore, exactly
\[
(Q-1)\,\tau(A_1,\ldots,A_t)
\]
non-zero codewords of $\calC'$ have full sum-rank weight $N-k+2$.
\end{remark}

The bound in \cref{thm:projection} depends on how the $k-2$ shortening positions are distributed among the blocks.
More precisely, shortening according to a vector $\abf\le\nbf$ reduces the total number of points of the associated projective system from $M(\nbf)$ to $M(\nbf-\abf)$. To obtain the strongest possible bound, we minimize this loss over all distributions of a fixed number of shortened dimensions. Accordingly, we define for $0\le j\le N$ the \emph{loss}
\[
\lambda_j(\nbf):=
\min\left\{
\sum_{i=1}^t
\bigl(\theta_{n_i-1}-\theta_{n_i-a_i-1}\bigr):
\ \zerobf\le\abf\le\nbf,\ |\abf|=j
\right\}.
\]

\begin{lemma}\label{lem:loss}
The following hold.
\begin{enumerate}[label=\textup{(\roman*)}]
\item The minimum defining $\lambda_j(\nbf)$ is attained at some $\abf$ with at most one index $i$ such that $0<a_i<n_i$.
\item $\lambda_j(\nbf)\geq j$, with equality if and only if at least $j$ blocks have rank one.
\item $\lambda_1(\nbf)=q^{r(\nbf)-1}$.
\item If $n_i\geq j$ for some $i$, then $\lambda_j(\nbf)\le q^{n_i-j}\,\theta_{j-1}$.
\end{enumerate}
\end{lemma}
\begin{proof}
For an integer $0\le a\le n$, let
\[
\Lambda_n(a):=\theta_{n-1}-\theta_{n-a-1}
       =\sum_{s=n-a}^{n-1}q^s,
\]
where an empty sum is zero. In particular, $\Lambda_n(0)=0$ and $\Lambda_n(n)=\theta_{n-1}$, since $\theta_{-1}=0$.

Thus $\Lambda_n(a)$ measures the loss associated with
shortening the code at $a$ positions from a block of rank $n$.
In particular, for $0\le a<n$, by \eqref{eq:theta-difference}, we have
\[
\Lambda_n(a+1)-\Lambda_n(a)=q^{n-a-1}.
\]

\smallskip
\noindent
(i) Let $\abf$ be a vector attaining the minimum
in the definition of $\lambda_j(\nbf)$.
Suppose, by contradiction, that there exist two
distinct indices $i,h$ such that $0<a_i<n_i$ and $0<a_h<n_h$.
Without loss of generality, assume that $n_i-a_i\le n_h-a_h$.

Consider the vector $\abf'=\abf+\ebf_i-\ebf_h$.
Since $a_i<n_i$ and $a_h>0$, the vector
$\abf'$ still satisfies
$\zerobf\le\abf'\le\nbf$ and
$|\abf'|=j$. Moreover, the difference between the corresponding losses is
\begin{align*}
&\sum_{s=1}^t \Lambda_{n_s}(a_s')
 -\sum_{s=1}^t \Lambda_{n_s}(a_s)\\
&\qquad =
\bigl(\Lambda_{n_i}(a_i+1)-\Lambda_{n_i}(a_i)\bigr)
+
\bigl(\Lambda_{n_h}(a_h-1)-\Lambda_{n_h}(a_h)\bigr)\\
&\qquad =
q^{n_i-a_i-1}-q^{n_h-a_h}<0,
\end{align*}
where the last inequality follows from
$n_i-a_i\le n_h-a_h$.

This contradicts the minimality of $\abf$.
Therefore, at most one component of a minimizer
can satisfy $0<a_i<n_i$.

\smallskip
\noindent
(ii) For $j=0$ the assertion is trivially satisfied, so assume $j>0$.
For every integer $0< a\le n$, we have
\[
\Lambda_n(a)=\sum_{s=n-a}^{n-1}q^s\ge a,
\]
since every summand is at least one. The above inequality is clearly satisfied when $a=0$.
Moreover, equality holds if and only if
$a=0$, or $n=a=1$.

Consequently, for every admissible $\abf$ with
$|\abf|=j$,
\[
\sum_{i=1}^t \Lambda_{n_i}(a_i)
\ge\sum_{i=1}^t a_i=j.
\]
Equality is possible precisely when every
non-zero component $a_i$ corresponds to a
block of rank one. Since $|\abf|=j$,
this is equivalent to the existence of at
least $j$ blocks of rank one.

\smallskip
\noindent
(iii) If $j=1$, the only admissible vectors
are the standard unit vectors $\ebf_i$.
Since $\Lambda_{n_i}(1)=q^{n_i-1}$,
we obtain $\lambda_1(\nbf)
=\min_{i}q^{n_i-1}
=q^{r(\nbf)-1}$.

\smallskip
\noindent
(iv) Suppose that $n_i\ge j$ for some $i$.
Choosing $\abf=j\ebf_i$ gives
\[
\lambda_j(\nbf)
\le \Lambda_{n_i}(j)
=\sum_{s=n_i-j}^{n_i-1}q^s
=q^{n_i-j}\theta_{j-1},
\]
as claimed.
\end{proof}

\begin{remark}
By \cref{lem:loss}(i), the computation of $\lambda_j(\nbf)$ can be restricted to vectors $\abf$ for which all but at most one component belong to $\{0,n_i\}$.
In other words, the minimum loss is attained by shortening entire blocks, except possibly for one block, which may be shortened only partially.
\end{remark}

\begin{corollary}\label{cor:profile}
If $\calU$ is a non-trivial $[\nbf,k]_{Q/q}$ MSRD system, then
\[
M(\nbf)\le Q+1+\lambda_{\kappa-2}(\nbf).
\]
In particular:
\begin{enumerate}[label=\textup{(\roman*)}]
\item if $\kappa=3$, then $M(\nbf)\le Q+1+q^{r(\nbf)-1}$;
\item if $\kappa\geq3$ and $n_i\geq\kappa-2$ for some $i$, then $M(\nbf)\le Q+1+q^{n_i-\kappa+2}\theta_{\kappa-3}$.
\end{enumerate}
\end{corollary}
\begin{proof}
Let $\calC$ be an MSRD code associated with $\calU$.
By \cref{lem:puncture-dual}(ii), both $\calC$ and $\calC^\perp$ are MSRD codes with the same profile $\nbf$. We may therefore apply \cref{thm:projection} to the code among $\calC$ and $\calC^\perp$ whose dimension is $\kappa$. Hence, for every vector $\zerobf\le\abf\le\nbf$ with $|\abf|=\kappa-2$, we have
\[
\sum_{i=1}^t\theta_{n_i-a_i-1}\le Q+1.
\]
By the definition of $M(\nbf)$, this can be rewritten as
\[
M(\nbf)\le Q+1+
\sum_{i=1}^t
\left(\theta_{n_i-1}-\theta_{n_i-a_i-1}\right).
\]
Since this inequality holds for every admissible $\abf$, taking the minimum over all such vectors gives $M(\nbf)\le Q+1+\lambda_{\kappa-2}(\nbf)$.

Finally, if $\kappa=3$, then
$\lambda_{\kappa-2}(\nbf)=\lambda_1(\nbf)
=q^{r(\nbf)-1}$ by \cref{lem:loss}(iii), which proves (i).

If $\kappa\ge3$ and $n_i\ge\kappa-2$ for some $i$, then \cref{lem:loss}(iv) gives
\[
\lambda_{\kappa-2}(\nbf)
\le q^{n_i-\kappa+2}\theta_{\kappa-3},
\]
and (ii) follows.
\end{proof}

\begin{corollary}\label{cor:length-general}
Let $\calU$ be a non-trivial $[\nbf,k]_{Q/q}$ MSRD system. Then the following hold.
\begin{enumerate}[label=\textup{(\roman*)}]
\item $N\le Q+\kappa-1$.
\item If there is an index $h$ such that $n_h\geq2$ and $N-n_h\geq\kappa-2$, then $N\le Q+\kappa-q$.
\item If at least $\kappa-2$ blocks have rank one, then $t+q(N-t)\le Q+\kappa-1$.
\item If $N=Q+\kappa-1$, then $\calU$ is classical.
\item If $t\geq Q+\kappa-q$, then $\calU$ is classical.
\end{enumerate}
\end{corollary}
\begin{proof}
Since $\calU$ is non-trivial, we have $2\le k\le N-2$, so that $\kappa\geq2$. By \cref{thm:projection}, applied directly if $\kappa=k$ and through its last statement if $\kappa=N-k$, inequality~\eqref{eq:projection} holds for every $\zerobf\le\abf\le\nbf$ with $|\abf|=\kappa-2$. For such an $\abf$, put
\[
\sigma(\abf):=\sum_{i=1}^t\max\{n_i-a_i-1,0\}.
\]
Applying~\eqref{eq:theta-lower} with $n=n_i-a_i$ for every $i$, and using $\sum_i(n_i-a_i)=N-\kappa+2$, we obtain
\begin{equation}\label{eq:length-general}
Q+1\geq\sum_{i=1}^t\theta_{n_i-a_i-1}\geq(N-\kappa+2)+(q-1)\,\sigma(\abf).
\end{equation}
\noindent(i) Since $\sigma(\abf)\geq0$, inequality~\eqref{eq:length-general} gives $N\le Q+\kappa-1$.
\noindent(ii) Since $\sum_{i\neq h}n_i=N-n_h\geq\kappa-2$, we can choose $\abf$ with $|\abf|=\kappa-2$ and $a_h=0$. Then $\sigma(\abf)\geq n_h-1\geq1$, and~\eqref{eq:length-general} gives $Q+1\geq N-\kappa+2+(q-1)$, that is, $N\le Q+\kappa-q$.
\noindent(iii) Let $I$ be a set of $\kappa-2$ indices such that $n_i=1$ for every $i\in I$, and put $\abf:=\sum_{i\in I}\ebf_i$. For $i\in I$ we have $\max\{n_i-a_i-1,0\}=0=n_i-1$, and for $i\notin I$ we have $\max\{n_i-a_i-1,0\}=n_i-1$. Hence $\sigma(\abf)=\sum_{i=1}^t(n_i-1)=N-t$, and~\eqref{eq:length-general} gives $N+(q-1)(N-t)\le Q+\kappa-1$, which is equivalent to $t+q(N-t)\le Q+\kappa-1$.
\noindent(iv) Suppose that $N=Q+\kappa-1$ and that $\calU$ is not classical, so that $n_h\geq2$ for some $h$. Since $n_h\le m<Q+1=N-\kappa+2$ (see \cref{rem:ni-le-m}), we have $N-n_h>\kappa-2$. Hence (ii) applies and gives $N\le Q+\kappa-q<Q+\kappa-1$, a contradiction.
\noindent(v) Suppose that $\calU$ is not classical, so that $n_h\geq2$ for some $h$. Since every block has rank at least one, we have $N-n_h\geq t-1\geq Q+\kappa-q-1\geq\kappa-2$. Hence (ii) applies and gives $N\le Q+\kappa-q$. On the other hand, $N\geq t+1\geq Q+\kappa-q+1$, a contradiction.
\end{proof}

\begin{remark}\label{rem:length-general}
\begin{enumerate}[label=\textup{(\alph*)}]
\item For $k=3$ and $N\geq6$, we have $\kappa=3$, and \cref{cor:length-general} recovers the results of \cref{sec:dim3}. Parts (i), (ii) and (iii) are exactly parts (i), (ii) and (iii) of \cref{thm:length-dim3}. Indeed, in (ii) the condition $N-n_h\geq1$ means that $t\geq2$, and in (iii) the hypothesis is that at least one block has rank one. Part (iv), combined with \cref{thm:plane-arcs}(ii), gives \cref{cor:extremal-dim3}(i). Part (v) states that $t\geq Q+3-q$ forces the system to be classical; \cref{thm:many-blocks} improves this to $t\geq Q+2-q$ by means of a parity argument.
\item Part (ii) improves the bound of (i) by $q-1$ as soon as some block has rank at least two and the remaining blocks have total length at least $\kappa-2$. Unlike \cref{cor:classical-transfer}(iii), this does not rely on the MDS conjecture. In particular, if moreover $\kappa\le q$, then (ii) gives $N\le Q$, which is stronger than the bound $N\le Q+1$ that the MDS conjecture would give through \cref{cor:classical-transfer}(iii).
\item Since $t\le N$, with equality if and only if $\calU$ is classical, part (i) also gives $t\le Q+\kappa-1$. Part (v) improves this for non-classical systems: if $\calU$ is not classical, then $t\le Q+\kappa-q-1$.
\end{enumerate}
\end{remark}

\subsection{Equal and full blocks}

We now specialize the preceding bounds to MSRD systems whose blocks have equal lengths, obtaining explicit restrictions on the number of blocks.
Of particular interest is the case of full-length blocks, where $n_i=m$ for every $i$.
In this setting, we recover known bounds and identify parameter ranges in which the maximum number of blocks is attained by linearized Reed--Solomon codes.

\begin{corollary}\label{cor:equal-blocks}
Let $\calU$ be a non-trivial MSRD system with profile $(s,\ldots,s)$ consisting of $t$ blocks, and write $\kappa-2=l s+\delta$ with $0\le\delta<s$. Then
\[
t\le l+1+\left\lfloor\frac{Q+1-\theta_{s-\delta-1}}{\theta_{s-1}}\right\rfloor.
\]
\end{corollary}
\begin{proof}
Since $\kappa-2<N=ts$, we have $t\geq l+1$. Remove $ l $ blocks entirely and $\delta$ dimensions from one further block. Then~\eqref{eq:projection} (applied to either $\calC$ or $\calC^\perp$) reads $(t- l -1)\theta_{s-1}+\theta_{s-\delta-1}\le Q+1$.
\end{proof}

We can refine the above bound in the case in which the profile of the code is $(m,\ldots,m)$.

\begin{corollary}\label{cor:full-blocks}
Let $\calU$ be a non-trivial MSRD system with profile $(m,\ldots,m)$ consisting of $t$ blocks. Then
\[
t\le q-1+\Bigl\lfloor\frac{\kappa-2}{m}\Bigr\rfloor+\epsilon,\qquad\text{where }\epsilon=\begin{cases}1&\text{if }m\mid\kappa-1,\\0&\text{otherwise.}\end{cases}
\]
In particular, $t\le q-1$ if $\kappa\le m$, and $t\le q$ if $\kappa=m+1$. 
\end{corollary}
\begin{proof}
Apply \cref{cor:equal-blocks} with $s=m$, and recall that $(q-1)\theta_{m-1}=Q-1$. If $\delta\le m-2$, then $q+1\le\theta_{m-\delta-1}\le\theta_{m-1}$, so that
\[
q-2<\frac{Q+1-\theta_{m-1}}{\theta_{m-1}}\le\frac{Q+1-\theta_{m-\delta-1}}{\theta_{m-1}}\le\frac{(q-1)(Q-q)}{Q-1}<q-1,
\]
and the floor equals $q-2$. If $\delta=m-1$, that is, if $m\mid\kappa-1$, then the numerator is $Q=(q-1)\theta_{m-1}+1$ and the floor equals $q-1$. If $\kappa\le m$, then $ l =0$ and $\delta=\kappa-2\le m-2$, if $\kappa=m+1$, then $ l =0$ and $\delta=m-1$.
\end{proof}

The bound $t\le q-1$ is attained in every dimension by linearized Reed--Solomon codes. 

\begin{remark}\label{rem:bgr}
\Cref{cor:equal-blocks} recovers \cite[Theorem~VI.12]{byrne2021fundamental} for the special case of MSRD codes with minimum distance $d$ such that $d-3=\kappa-2$. 
\end{remark}

\subsection{The exceptional full-block case}
By \cref{cor:full-blocks}, an MSRD system with full blocks satisfies $t\le q-1$ whenever $\kappa\le m$, whereas for $\kappa=m+1$ the bound becomes $t\le q$.
Thus, in the range $\kappa\le m+1$, the only case in which the number of blocks could exceed $q-1$ is $t=q$ and $\kappa=m+1$.

Observe that this case is impossible when $q=2$.
Indeed, $t=q=2$ would imply $N=2m$, and hence
\[
\kappa=\min\{k,N-k\}\le\frac{N}{2}=m,
\]
a contradiction. For $q\ge3$, however, this restriction no longer applies, since $N=qm\ge2(m+1)$ for $m\ge2$.
We are therefore led to the following existence problem.

\begin{problem}\label{prob:exceptional}
Let $q\geq3$ and $m\geq2$. Is there an MSRD system of dimension $m+1$ over $\F_{q^m}$ with profile $(m,\ldots,m)$ consisting of $q$ blocks? Equivalently, by \cref{lem:puncture-dual}(ii), is there one of dimension $(q-1)m-1$? For $m=2$, this asks for a three-dimensional MSRD code over $\F_{q^2}$ with $q$ blocks of length $2$.
\end{problem}

\begin{proposition}\label{prop:exceptional-structure}
Suppose that a system $\calU$ as in \cref{prob:exceptional} exists. Then, for every block $U_i$ and every $(m-1)$-dimensional $\Fq$-subspace $A\le U_i$, exactly one hyperplane through $\langle A\rangle_{\FQ}$ has weight $m-1$. Equivalently, projecting from $\langle A\rangle_{\FQ}$ partitions $\PG(1,Q)$ minus one point into $q-1$ scattered linear sets of rank $m$ and one further point. For $m=2$, this means that every point of $\calL_{U_1}\cup\cdots\cup\calL_{U_q}$ lies on exactly one tangent line.
\end{proposition}
\begin{proof}
Fix an index $i$ and an $(m-1)$-dimensional
$\Fq$-subspace $A\le U_i$, and put $W:=\langle A\rangle_{\FQ}$.
Since $k=m+1$, we have $\dim_{\FQ}W=m-1=k-2$ by \cref{lem:projection-shortening}.

Applying \cref{thm:projection} with $\abf=(m-1)\ebf_i$, we obtain a two-dimensional MSRD system with profile
\[
(1,m,\ldots,m),
\]
consisting of one block of rank one and $q-1$ blocks of rank $m$.

By \cref{prop:dim2}, the associated linear sets are scattered and pairwise disjoint.
Their union therefore contains exactly
\[
\theta_0+(q-1)\theta_{m-1}
=1+(Q-1)=Q
\]
points of $\PG(1,Q)$.

Since $\PG(1,Q)$ has $Q+1$ points, exactly one point lies outside this union.
Equivalently, the projected system has exactly one point of weight zero.

By \cref{lem:projection-shortening}(ii), the weight of a hyperplane $H$ containing $W$ is related to the weight of the corresponding point $H/W$ by
\[
\wt_\calU(H)
=
\wt_{\calU/W}(H/W)+m-1.
\]
Consequently, exactly one hyperplane through $W$ has weight $m-1$, while all the others have weight $m$.
\end{proof}

The following example shows that \cref{prob:exceptional} has a positive answer for $(q,m)=(3,2)$.

\begin{theorem}\label{thm:sporadic}
Let $q=3$, $Q=9$, and let $a\in\F_9$ with $a^4=-1$. Then
\[
U_0:=\{(0,y,ay^3):y\in\F_9\},\qquad U_1:=\{(ay^3,0,y):y\in\F_9\},\qquad U_2:=\{(y,ay^3,0):y\in\F_9\}
\]
form a $[(2,2,2),3]_{9/3}$ MSRD system. Equivalently, if $f_{b,c}(y):=by+ca\,y^3$ and $\beta_1,\beta_2$ is an $\F_3$-basis of $\F_9$, then
\[
\bigl\{\bigl(f_{x_1,x_2}(\beta_1),f_{x_1,x_2}(\beta_2)\ \big|\ f_{x_2,x_0}(\beta_1),f_{x_2,x_0}(\beta_2)\ \big|\ f_{x_0,x_1}(\beta_1),f_{x_0,x_1}(\beta_2)\bigr):(x_0,x_1,x_2)\in\F_9^3\bigr\}
\]
is a $[(2,2,2),3,4]_{9/3}$ MSRD code with $t=q=3$ blocks of length $m=2$.
\end{theorem}
\begin{proof}
Let $\nu:=a^{q+1}=\Norm_{\F_9/\F_3}(a)=-1$. For $\xbf=(x_0,x_1,x_2)\in\F_9^3$ and $\ubf=(0,y,ay^3)\in U_0$ we have $\xbf\cdot\ubf=f_{x_1,x_2}(y)$, so that $\wt_{U_0}(\xbf^\perp)=\dim_{\F_3}\ker f_{x_1,x_2}$, similarly, $U_1$ corresponds to $f_{x_2,x_0}$ and $U_2$ to $f_{x_0,x_1}$. In particular, the vectors of $U_0$, $U_1$ and $U_2$ obtained for $y\in\{\beta_1,\beta_2\}$ are the columns of a generator matrix of the code in the statement.

\emph{Claim:} $\dim_{\F_3}\ker f_{b,c}$ equals $2$ if $b=c=0$, it equals $1$ if $bc\neq0$ and $\Norm(b)=\nu\Norm(c)$, and it equals $0$ otherwise. Indeed, if exactly one of $b,c$ is zero, then $f_{b,c}$ is bijective. If $bc\neq0$, a non-zero root of $f_{b,c}$ satisfies $y^{q-1}=-b/(ca)$. Since $y\mapsto y^{q-1}$ maps $\F_{q^2}^*$ onto the elements of norm one, with fibers of the form $y\Fq^*$, such roots exist if and only if $\Norm(-b/(ca))=1$, that is, if and only if $\Norm(b)=\nu\Norm(c)$ (recall that $\Norm(-1)=1$), and in this case they form a one-dimensional $\Fq$-subspace minus $\zerobf$.

Now let $\xbf\neq\zerobf$. If two coordinates of $\xbf$ vanish, then $\wt_\calU(\xbf^\perp)=2+0+0=2$. If exactly one coordinate vanishes, then at most one of the three kernels is non-zero, and it has dimension one. If no coordinate vanishes, then $\wt_\calU(\xbf^\perp)$ is the number of true equalities among
\[
\Norm(x_1)=\nu\Norm(x_2),\qquad\Norm(x_2)=\nu\Norm(x_0),\qquad\Norm(x_0)=\nu\Norm(x_1).
\]
All three together imply $\nu^3=1$, which is false since $\nu=-1$. Hence every hyperplane has weight at most two. Finally, the three subspaces span $\F_9^3$, since a hyperplane containing all of them would have weight $6$.
\end{proof}

\begin{remark}\label{rem:sporadic-general}
The proof of \cref{thm:sporadic} works verbatim over $\F_{q^2}$ for any $q\notin\{2,4\}$, with $f_{b,c}(y)=by+ca\,y^q$, as long as $\nu=\Norm(a)$ satisfies $\nu^3\neq1$ (such $a$ exists whenever $q-1\nmid 3$, \emph{i.e.} $q \ne 2,4$), it yields MSRD systems with profile $(2,2,2)$. Only for $q=3$ does this exceed the $q-1$ blocks of linearized Reed--Solomon codes. Note also that the lines spanned by the three blocks are the three coordinate lines, which are not concurrent.
\end{remark}

\subsection{An MSRD conjecture}

The preceding results provide several restrictions on the possible profiles of MSRD systems, but a general characterization of the admissible parameters remains open.
Motivated by the classical MDS conjecture and the geometric interpretation of MSRD codes, we propose a conjectural upper bound on $M(\nbf)$ that extends the classical setting.
Before stating it, we introduce the notion of realisability and recall some consequences of shortening, puncturing,
and duality.

\begin{definition}\label{def:realisable}
We say that $(\nbf,k)$ is \emph{realizable over $\F_{q^m}$} if there exists an $[\nbf,k]_{Q/q}$ MSRD system, up to reordering the blocks.
\end{definition}

By \cref{lem:projection-shortening,lem:puncture-dual,prop:dim2}, realizability has the following closure properties:
\begin{itemize}
\item it is invariant under $k\mapsto N-k$, for $k\le N-1$;
\item if $(\nbf,k)$ is realizable, then so is $(\nbf-\abf,k-|\abf|)$ for every $\abf\le\nbf$ with $|\abf|\le k-2$;
\item if $(\nbf,k)$ is realizable and $N\geq k+1$, then so is $(\nbf-\ebf_i,k)$ for every $i$;
\item for $k=2$, it is a packing problem for scattered linear sets on $\PG(1,Q)$, and it forces $M(\nbf)\le Q+1$.
\end{itemize}

For classical profiles, the cases $\kappa\le3$ are settled: $\kappa=2$ is \cref{prop:dim2} together with duality, and $\kappa=3$ is \cref{thm:plane-arcs} together with duality. Moreover, non-trivial MDS codes over $\FQ$ have $k\le Q-1$ and $N-k\le Q-1$, by \cref{thm:bush} or by \cref{thm:projection} applied to $\calC$ and $\calC^\perp$. The exceptional parameters of \cref{conj:mds} ($Q$ even, $k\in\{3,Q-1\}$ and $N=Q+2$) are exactly those with $\kappa=3$. Hence the MDS conjecture over $\FQ$ is equivalent to the statement that
\[
\textit{every non-trivial classical MSRD system with }\kappa\geq4\textit{ has }N\le Q+1.
\]
In other words, the bound $M(\nbf)\le Q+1$, which holds in dimension two, is expected to persist in every dimension with $\kappa\geq4$ for classical profiles. This suggests the following.

\begin{conjecture}[MSRD conjecture]\label{conj:msrd}
Let $(\nbf,k)$ be realisable over $\F_{q^m}$, with $1\le n_i\le m$ for every $i$ and $\kappa=\min\{k,N-k\}\geq4$. Then
\[
M(\nbf)=\sum_{i=1}^t\theta_{n_i-1}\le Q+1.
\]
\end{conjecture}

\begin{remark}\label{rem:msrd-conj}
\begin{enumerate}[label=\textup{(\alph*)}]
\item For $\nbf=(1,\ldots,1)$, \cref{conj:msrd} is the MDS conjecture. Compared with \cref{cor:profile}, it asserts that the loss term $\lambda_{\kappa-2}(\nbf)$ can be dropped once $\kappa\geq4$.
\item For full blocks, it says that $t\le q-1$, since $Q+1<q\,\theta_{m-1}$. Thus linearized Reed--Solomon codes would have the largest possible number of full blocks, exactly as Reed--Solomon codes are conjecturally the longest MDS codes.
\item Equality is attained by extended Reed--Solomon codes of every dimension (the classical case), and in dimension two by the profile $(m,\ldots,m,1,1)$ of \cref{prop:dim2}.
\item The hypothesis $\kappa\geq4$ cannot be weakened: for $\kappa=3$ there are the hyperovals, with $M(\nbf)=Q+2$ for $Q$ even, and the system of \cref{thm:sporadic}, with $M(\nbf)=12=Q+3$.
\end{enumerate}
\end{remark}

\section{Conclusions}\label{sec:conclusions}

In this paper, we have investigated the length problem for MSRD codes through their geometric description in terms of systems of $\Fq$-subspaces. Our results show that, although every MSRD code is necessarily MDS with respect to the Hamming metric over $\FQ$, the MSRD condition imposes additional geometric restrictions depending on the block profile.

In dimension three, we have determined the possible intersection sizes of lines with the associated projective MSRD systems and obtained their character distributions.
These results yield explicit bounds on the total length, as well as statements showing that sufficiently long MSRD codes must be classical.
These results yield explicit bounds on the total length, and show that MSRD codes that are sufficiently long, or that have sufficiently many blocks, must be classical.
We have also proved the non-existence of three-dimensional one-weight MSRD codes. In arbitrary dimension, the geometric interpretation of shortening as projection provides a unified approach to length bounds, extending classical counting arguments for arcs and yielding restrictions that depend explicitly on the block profile.

The bounds obtained for full-length blocks are particularly interesting. They show that the number of such blocks cannot exceed $q-1$ whenever $\kappa\le m$, in agreement with the parameters attained by linearized Reed--Solomon codes. On the other hand, the explicit construction of \cref{thm:sporadic} shows that exceptional configurations may occur. Understanding whether further exceptions exist, particularly in the parameter range considered in \cref{prob:exceptional}, remains an interesting problem.
More generally, \cref{conj:msrd} proposes a sum-rank analogue of the MDS conjecture, suggesting that $M(\nbf)\le Q+1$ whenever $\kappa\ge4$.

\section*{Statements and Declarations}

\noindent\textbf{Funding.} The work of the first author was funded by the Agence de l'innovation de d\'efense (AID), France. The research of the second author was partially supported by the Gruppo Nazionale per le Strutture Algebriche, Geometriche e le loro Applicazioni (GNSAGA) of the Istituto Nazionale di Alta Matematica (INdAM).

\noindent\textbf{Competing interests.} The authors have no competing interests to declare that are relevant to the content of this article.

\noindent\textbf{Data availability.} Data sharing is not applicable to this article, as no datasets were generated or analyzed during the current study.

\section*{Declaration of Generative AI and AI-assisted Technologies in the Writing Process}
The authors used AI-assisted tools exclusively for language editing and improving the clarity and readability of the manuscript. All mathematical results, proofs, and scientific conclusions are the authors' own.
The authors have carefully reviewed the entire manuscript and take full responsibility for its content.

\bibliographystyle{abbrv}
\bibliography{biblio}

\end{document}